\documentclass[aos,preprint]{imsart}

\RequirePackage[OT1]{fontenc}
\RequirePackage{amsthm,amsmath,amsfonts,amssymb}
\RequirePackage[numbers]{natbib}
\RequirePackage[colorlinks,citecolor=blue,urlcolor=blue]{hyperref}
\RequirePackage{graphicx}
\usepackage[table,xcdraw]{xcolor}

\startlocaldefs
\usepackage{graphicx,mathrsfs,epstopdf}
\usepackage{amsmath,amsfonts,amssymb,amsthm}
\usepackage[colorlinks,citecolor=blue,urlcolor=blue]{hyperref}
\usepackage{algorithm} 
\usepackage{algorithmic}
\usepackage{framed} 
\usepackage{blkarray}

\usepackage{enumerate}
\usepackage{subfigure}
\usepackage{tikz}
\usetikzlibrary{positioning}
\newtheorem{thm}{Theorem}[section]
\newtheorem{lem}[thm]{Lemma}
\newtheorem{prop}[thm]{Proposition}
\newtheorem{cor}[thm]{Corollary}

\theoremstyle{definition}

\newtheorem{ex}[thm]{Example}

\theoremstyle{remark}
\newtheorem{rem}[thm]{Remark}

\newcommand{\N}{\mathbb{N}}

\newcommand{\R}{\mathbb{R}}

\newcommand{\cX}{\mathcal{X}}

\newcommand{\mtp}{{\rm MTP}_{2}}
\newcommand{\mwst}{{\rm MWST}}
\newcommand{\mwsf}{{\rm MWSF}}
\newcommand{\ec}{{\rm EC}}

\newcommand{\cd}{\,|\,}
\newcommand{\cip}{\mbox{\,$\perp\!\!\!\perp$\,}}

\def\proofname{\indent {\sl Proof}}

\newcommand{\add}[1]{#1}
\newcommand{\del}{}

\endlocaldefs

\endlocaldefs

\begin{document}

\begin{frontmatter}
\title{Maximum likelihood estimation in Gaussian models under total positivity}
\runtitle{Maximum likelihood in totally positive Gaussian models}

\begin{aug}
\author{\fnms{Steffen} \snm{Lauritzen}\thanksref{t2}\ead[label=e1]{lauritzen@math.ku.dk}},
\author{\fnms{Caroline} \snm{Uhler}\thanksref{t1}\ead[label=e2]{cuhler@mit.edu}},
\and
\author{\fnms{Piotr} \snm{Zwiernik}\thanksref{t3}\ead[label=e3]{piotr.zwiernik@upf.edu}}

\runauthor{S.~Lauritzen, C.~Uhler, P.~Zwiernik}

\affiliation{University of Copenhagen\thanksmark{t1}, 
Massachusetts Institute of Technology\thanksmark{t2}, 
and Universitat Pompeu Fabra\thanksmark{t3}}

\address{S.~Lauritzen\\
Department of Mathematical Sciences\\
University of Copenhagen\\
Copenhagen, Denmark\\
\printead{e1}}

\address{C.~Uhler\\
Laboratory for Information and Decision Systems\\
and Institute for Data, Systems and Society\\
Massachusetts Institute of Technology\\Cambridge, MA, U.S.A.\\
\printead{e2}}

\address{P.~Zwiernik\\
Department of Economics and Business\\
Universitat Pompeu Fabra\\
Barcelona, Spain\\
\printead{e3}}
\end{aug}

\begin{abstract}
We analyze the problem of maximum likelihood estimation for Gaussian distributions that are multivariate totally positive of order two ($\mtp$). By exploiting connections to phylogenetics and single-linkage clustering, we give a simple proof that the maximum likelihood estimator (MLE) for such distributions exists based on $n\geq 2$ observations, irrespective of the underlying dimension. Slawski and Hein~\cite{slawski2015estimation}, who first proved this result, also provided empirical evidence showing that the $\mtp$ constraint serves as an implicit regularizer and leads to sparsity in the estimated inverse covariance matrix, determining what we name the ML graph.  \del
\add{We} show that we can find an upper bound for the ML graph by adding edges corresponding to correlations in excess of those explained by the maximum weight spanning forest of the correlation matrix. \del 
\add{Moreover, we}  provide globally convergent coordinate descent algorithms for calculating the MLE under the $\mtp$ constraint which are  structurally similar to iterative proportional scaling. We conclude the paper with a discussion of signed $\mtp$ distributions.
\end{abstract}



\end{frontmatter}

\section{Introduction}

Total positivity is a special form of positive dependence between random variables that became an important concept in modern statistics; see, e.g., \cite{forcina2000,colangelo2005some,KarlinRinott80}. This property (also called the $\mtp$ property) appeared in the study of stochastic orderings, asymptotic statistics, and in statistical physics \cite{fortuin1971correlation,newman1983general}. Families of distributions with this property lead to many computational advantages \cite{bartolucci2002recursive,djolonga2015scalable,propp1996exact}.  In a recent paper~\cite{MTP2Markov2015}, the $\mtp$ property was studied in the context of graphical models and conditional independence in general. It was shown that $\mtp$ distributions have desirable Markov properties. 
Our paper can be seen as a continuation of this work with a focus on Gaussian distributions. 

A $p$-variate real-valued distribution with density $f$ w.r.t.\ a product measure $\mu$ is \emph{multivariate totally positive of order 2} ($\mtp$) if the density satisfies
\[f(x)f(y)\leq f(x\wedge y)f(x\vee y).\]
A multivariate Gaussian distribution with mean $\mu$ and a positive definite covariance matrix $\Sigma$ is $\mtp$ if and only if the concentration matrix $K:=\Sigma^{-1}$ is a symmetric \emph{M-matrix}, that is, $K_{ij}\leq 0$ for all $i\neq j$ or, equivalently, if all partial correlations are nonnegative. Such distributions were considered by B{\o{lviken}~\cite{B} and Karlin and Rinott~\cite{karlinGaussian}. Moreover, Gaussian graphical models, or Gaussian Markov random fields, were studied in the context of totally positive distributions in~\cite{malioutov2006walk}. $\mtp$ Gaussian graphical models were shown to form a sub-class of \emph{non-frustrated} Gaussian graphical models, which themselves are a sub-class of \emph{walk-summable} Gaussian graphical models. Efficient structure estimation algorithms for  $\mtp$ Gaussian graphical models were given in~\cite{anandkumar2012high} based on thresholding covariances after conditioning on subsets of variables of limited size. 
Efficient learning procedures based on convex optimization were suggested by Slawski and Hein~\cite{slawski2015estimation} and this paper is closely related to their approach; see also \cite{bhattacharya2012covariance} and \cite{egilmez:etal:16}. 

Throughout this paper, we assume that we are given $n$ i.i.d.~samples from $\mathcal{N}(\mu, \Sigma)$, where $\Sigma$ is an unknown positive definite matrix of size $p \times p$. Without loss of generality, we assume that $\mu=0$ and we focus on the estimation of $\Sigma$. We denote the sample covariance matrix based on $n$ samples by $S$.~Then~the log-likelihood function is, up to additive and multiplicative constants, given~by
\begin{equation}\label{eq:likely}
	\ell(K;S)\;\;=\;\;\log\det K-{\rm tr}(SK).
\end{equation}
 We denote the  cone of real symmetric matrices of size $p\times p$ by $\mathbb{S}^p$, its positive definite elements by $\mathbb{S}^p_{\succ 0}$, and its positive semidefinite elements  by $\mathbb{S}^p_{\succeq 0}$. Note that $\ell(K;S)$ is a strictly concave function of $K\in\mathbb{S}^p_{\succeq 0}$ . Since M-matrices form a convex subset of $\mathbb{S}^p_{\succ 0}$, the optimization problem for computing the \emph{maximum likelihood estimator} (MLE) for $\mtp$ Gaussian models is a convex optimization problem. Slawski and Hein~\cite{slawski2015estimation} showed that the MLE exists, i.e., the global maximum of this optimization problem is attained, when $n\geq 2$. This yields a drastic reduction from $n\geq p$ without the $\mtp$ constraint. In addition, they provided empirical evidence showing that the $\mtp$ constraint serves as an implicit regularizer and leads to sparsity in the concentration matrix $K$.

In this paper, we analyze the sparsity pattern of the MLE $\hat{K}$ under the $\mtp$ constraint. For a $p\times p$ matrix $K$ we let $G(K)$ denote the undirected graph on $p$ nodes with an edge $ij$ if and only if $K_{ij}\neq 0$. \del 
\add{In} Proposition~\ref{prop:PP} we obtain a simple upper bound for the ML graph $G(\hat{K})$ by adding edges to the \add{smallest \emph{maximum weight spanning forest} (MWSF)} \del
corresponding to empirical correlations in excess of those provided by the MWSF. We illustrate \add{the problem}
\del 
in the following example.

 \begin{ex}\label{ex:intro}
We consider the \texttt{carcass} data that are discussed in~\cite{GraphModelR} and can be found in the \verb+R+-library \verb+gRbase+. This data set contains measurements of the thickness of meat and fat layers at different locations on the back of a slaughter pig together with the lean meat percentage on each of 344 carcasses. For our analysis we ignore the lean meat percentage, since, by definition, this variable should be negatively correlated with fat and positively correlated with meat so the joint distribution is unlikely to be $\mtp$. 
The sample correlation matrix $R$ for these data is 
$$\small
R=\begin{blockarray}{ccccccc}
\textrm{Fat11} & \textrm{Meat11} & \textrm{Fat12} & \textrm{Meat12} & \textrm{Fat13}&\textrm{Meat13} \\
\begin{block}{(rrrrrr)l}
1.00 & 0.04 & 0.84 & 0.08 & 0.82 & -0.03 &\textrm{Fat11}\\ 
0.04 & 1.00 & 0.04 & 0.87 & 0.13 & 0.86  & \textrm{Meat11} \\
0.84 & 0.04 & 1.00 & 0.01 & 0.83 & -0.03 &\textrm{Fat12}\\
0.08 & 0.87 & 0.01 & 1.00 & 0.11 & 0.90  & \textrm{Meat12} \\
0.82 & 0.13 & 0.83 & 0.11 & 1.00 & 0.02  &\textrm{Fat13} \\
-0.03 & 0.86 & -0.03 & 0.90 & 0.02 & 1.00 &\textrm{Meat13} \\
\end{block}
\end{blockarray}
$$
and its inverse, scaled to have diagonal elements equal to one, $\tilde K$, is
$$\small
\tilde K=\begin{blockarray}{ccccccc}
\textrm{Fat11} & \textrm{Meat11} & \textrm{Fat12} & \textrm{Meat12} & \textrm{Fat13}&\textrm{Meat13} \\
\begin{block}{(rrrrrr)l}
1.00 & \textcolor{red}{0.16} & -0.52 & -0.31 & -0.40 & \textcolor{red}{0.19}&\textrm{Fat11}\\
\textcolor{red}{0.16} & 1.00 & -0.05 & -0.42 & -0.17 & -0.37& \textrm{Meat11} \\
-0.52& -0.05 & 1.00 & \textcolor{red}{0.25} & -0.45 & -0.17&\textrm{Fat12}\\
-0.31 & -0.42 & \textcolor{red}{0.25} & 1.00 & -0.02 & -0.61& \textrm{Meat12} \\
-0.40 & -0.17 & -0.45 & -0.02 & 1.00 & \textcolor{red}{0.10}&\textrm{Fat13} \\
\textcolor{red}{0.19} & -0.37 & -0.17 & -0.61& \textcolor{red}{0.10} & 1.00&\textrm{Meat13} \\
\end{block}
\end{blockarray}
$$

Note that the off-diagonal entries of $\tilde K$ are the negative empirical partial correlations.
This sample distribution  is not $\mtp$; the positive entries in $\tilde K$ are highlighted in red. The MLE under $\mtp$ can be computed for example using \verb+cvx+~\cite{cvx} in \verb+matlab+ or using one of the simple coordinate descent algorithms discussed in Section~\ref{sec:duality}. In this particular example the MLE can  also be obtained through the explicit formula (\ref{eq:W}) in Section~\ref{sec:MLgraph}. The MLE of the correlation matrix, rounded to 2 decimals, is
$$\small
\hat{R}=\begin{blockarray}{ccccccc}
\textrm{Fat11} & \textrm{Meat11} & \textrm{Fat12} & \textrm{Meat12} & \textrm{Fat13}&\textrm{Meat13} \\
\begin{block}{(rrrrrr)l}
1.00 & \textcolor{blue}{0.10} & 0.84 & \textcolor{blue}{0.09} & 0.82 & \textcolor{blue}{0.09}&\textrm{Fat11}\\ 
\textcolor{blue}{0.10} & 1.00 & \textcolor{blue}{0.11} & 0.87 & 0.13 & 0.86 & \textrm{Meat11} \\
0.84 & \textcolor{blue}{0.11} & 1.00 & \textcolor{blue}{0.09} & 0.83 & \textcolor{blue}{0.09} & \textrm{Fat12}\\ 
\textcolor{blue}{0.09} & 0.87 & \textcolor{blue}{0.09} & 1.00 & \textcolor{blue}{0.11} & 0.90 & \textrm{Meat12} \\
0.82 & 0.13 & 0.83 & \textcolor{blue}{0.11} & 1.00 & \textcolor{blue}{0.11} &\textrm{Fat13} \\
\textcolor{blue}{0.09} & 0.86 & \textcolor{blue}{0.09} & 0.90 & \textcolor{blue}{0.11} & 1.00 &\textrm{Meat13} \\
\end{block}
\end{blockarray}
$$
The entries of $\hat{R}$ that changed compared to the sample correlation matrix $R$ are highlighted in blue\footnote{We note that $\hat\Sigma_{45}>S_{45}$; the entries appear equal only because of the 2-digit rounding.}. The sparsity pattern of $\hat{K}=\hat\Sigma^{-1}$ is captured by the ML graph $G(\hat K)$ shown in Figure~\ref{fig:carcass_MTP2}.  

\begin{figure}
\centering
\includegraphics[width=4cm]{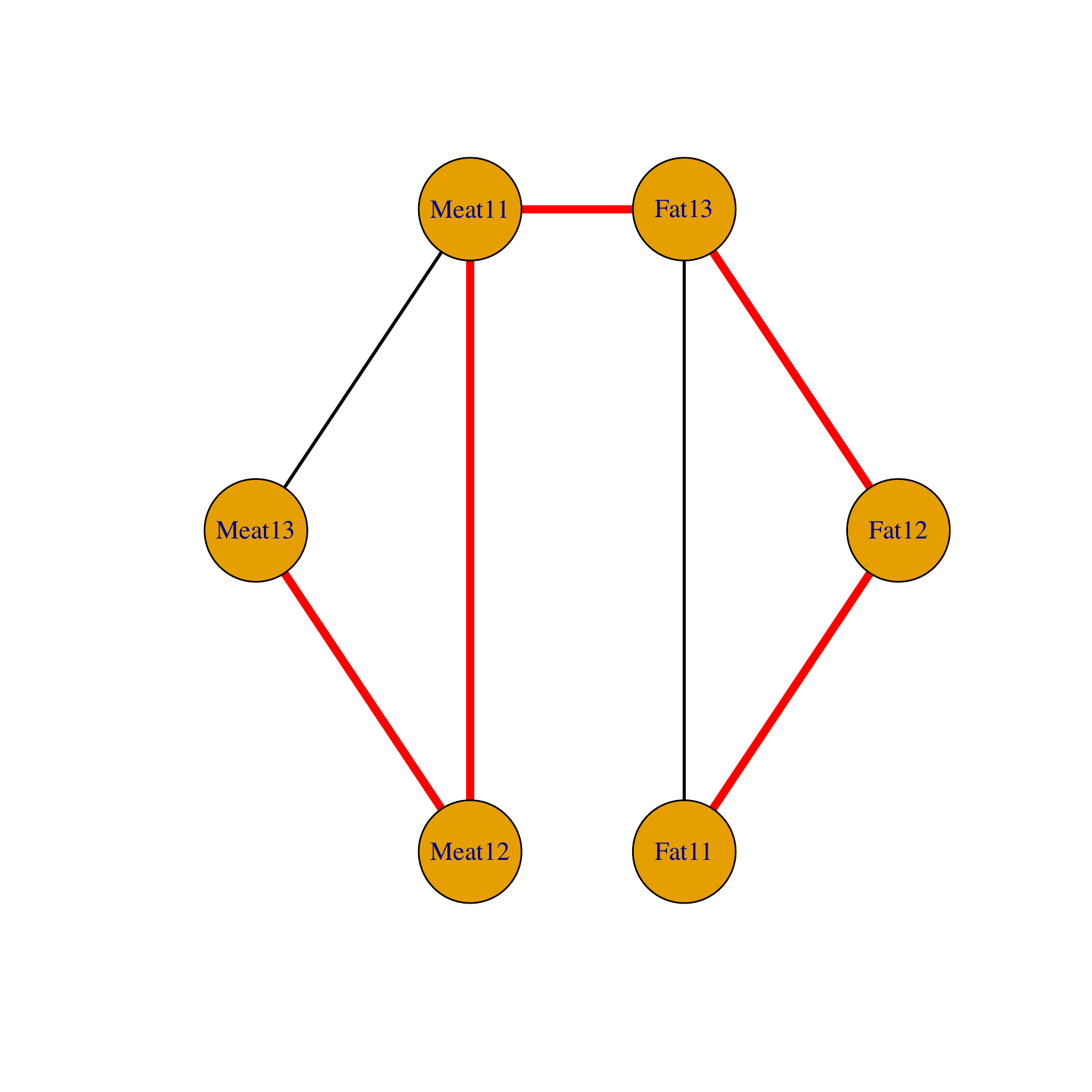}
\caption{\label{fig:carcass_MTP2}Undirected Gaussian graphical model for the \texttt{carcass} data obtained by estimating under the $\mtp$ assumtption. The thick red edges correspond to the MWSF of the correlation matrix.}

\end{figure}

Note that all edges corresponding to blue entries in $\hat{R}$ are missing in this graph. As we show in Proposition \ref{prop:mtp2GaussGraphModel}, this is a consequence of the KKT conditions. Consider now the maximum weight spanning forest of the complete graph with weights given by the entries of $R$. In this example, the spanning forest is a chain represented by the thick red edges in Figure~\ref{fig:carcass_MTP2}. By \add{Corollary~\ref{cor:blockest}} \del 
these edges form a spanning tree of the ML graph $G(\hat K)$. 

Interestingly, applying various methods for model selection such as stepwise AIC, BIC, or graphical lasso all yield  similar graphs, possibly indicating that the $\mtp$ assumption is quite reasonable. 
\end{ex}

The remainder of this paper is organized as follows: In Section~\ref{sec:duality}, we review the duality theory that is known more generally for regular exponential families and specialize it to $\mtp$ Gaussian distributions. This embeds the results by Slawski and Hein~\cite{slawski2015estimation} into the framework of exponential families and also leads to two related coordinate descent algorithms for computing the MLE, one that acts on the entries of $K$ and one that acts on the entries of $\Sigma$. 
In Section~\ref{sec:ultrametric}, we show how the problem of ML estimation for $\mtp$ Gaussian distributions is connected to single-linkage clustering and ultrametrics as studied in phylogenetics. These observations result in a simple proof of the existence of the MLE for $n\geq 2$, a result that was first proven in~\cite{slawski2015estimation}. Our proof is by constructing a primal and dual feasible point of the convex ML estimation problem for $\mtp$ Gaussian models. In Section~\ref{sec:MLgraph} \del
\add{we investigate the structure}
of the ML graph $G(\hat{K})$ 
and give a simple upper bound for it. Finally, in Section~\ref{sec:signed} we discuss how our results can be generalized to so-called \emph{signed} $\mtp$ Gaussian distributions, where the distribution is $\mtp$ up to sign changes or, equivalently,  $|X|$  is $\mtp$. Such distributions were introduced by Karlin and Rinott in~\cite{karlin1981signed}. We conclude the paper with a brief discussion of various open problems.

\section{Duality theory for ML estimation under $\mtp$}
\label{sec:duality}

We start this section by formally introducing absolutely continuous $\mtp$ distributions and then discuss the duality theory for Gaussian $\mtp$ distributions. Let $V:=\{1,2,\dots , p\}$ be a finite set and let $X=(X_i, i\in V)$ be a  random vector with density $f$ w.r.t.\ Lebesgue measure on the product space $\cX=\prod_{i\in V}\cX_{i}$, where $\cX_i\subseteq \R$ is the state space of $X_i$. We define the coordinate-wise minimum and maximum as
$$x\wedge y=(\min(x_{v},y_{v}),v\in V),\quad x\vee y=(\max(x_{i},y_{i}),i\in V).$$
Then we say that $X$ or the distribution of $X$ is \emph{multivariate totally positive of order two} ($\mtp$) if its density function $f$ on $\cX$ satisfies
\begin{equation}\label{eq:MTP2}
f(x)f(y)\quad\leq \quad f(x\wedge y)f(x\vee y)\qquad\mbox{for all }x,y\in \cX.
\end{equation}

In this paper, we concentrate on the Gaussian setting. It is easy to show that a Gaussian distribution with mean $\mu$ and covariance matrix $\Sigma$ is $\mtp$ if and only if $K=\Sigma^{-1}$ is a symmetric \emph{M-matrix}, i.e.\ $K$ is positive definite and
\begin{itemize}
\item[(i)] $K_{ii}>0$ for all $i\in V$,
\item[(ii)] $K_{ij}\leq 0$ for all $i,j\in V$ with $i\neq j$.
\end{itemize}
Properties of M-matrices were studied by Ostrowski~\cite{O} who chose the name to honor H.~Minkowski. The connection to multivariate Gaussian distributions was established by B{\o}lviken~\cite{B} and Karlin and Rinott~\cite{karlinGaussian}.

We denote the set of all symmetric M-matrices of size $p\times p$ by $\mathcal{M}^p$. Note that $\mathcal{M}^p$ is a convex cone. In fact, it is obtained by intersecting the positive definite cone $\mathbb{S}^p_{\succ 0}$ with all the coordinate half-spaces 
$$\mathcal{H}^p_{ij} = \{X\in\mathbb{S}^{p} \mid X_{ij}\leq 0\}$$
with $i\neq j$. For a convex cone $\mathcal{C}$ we denote its closure by $\overline{\mathcal{C}}$. 
Then $\overline{\mathcal{M}^p}$ is given by $\mathbb{S}^p_{\succeq 0}\cap_{i< j} \mathcal{H}^p_{ij}$ and the ML estimation problem for Gaussian $\mtp$ models can be formulated as the following optimization problem:
\begin{equation}
\label{ML_mtp2}
\begin{aligned}
& \underset{K}{\text{maximize}}
& & \log\det(K) - \textrm{trace}(KS)  \\
& \text{subject to}
&& K\in\mathcal{M}^p \\
\end{aligned}
\end{equation}
This is a convex optimization problem, since the objective function is concave on $\mathbb{S}^p_{\succeq 0}$.

Next, we introduce a second convex cone $\mathcal{N}^p$ that plays an important role for ML estimation in Gaussian $\mtp$ models. To formally define this cone, we introduce two partial orders on matrices. Let $A,B$ be two $p\times p$ matrices. Then $A\geq B$ means that $A_{ij}\geq B_{ij}$ for all $(i,j)\in V\times V$, and $A\succeq B$ means that $A-B\in\mathbb{S}^p_{\succeq 0}$. Then the cone $\mathcal{N}^p$ is defined as the negative closure of $\mathbb{S}^p_{\succ 0}$, i.e.
$$\mathcal{N}^p = \{X\in\mathbb{S}^{p}\mid \exists Y\in\mathbb{S}^p_{\succ 0} \textrm{ with } X\leq Y \textrm{ and } \textrm{diag}(X)=\textrm{diag}(Y)\}.$$
To simplify notation, we will suppress the dependence on $p$ and write $\mathbb{S}$,  $\mathbb{S}_{\succeq 0}$,  $\mathbb{S}_{\succ 0}$, $\mathcal{M}$ and $\mathcal{N}$, when the dimension is clear. In the following result, we show that the cones $\mathcal{N}$ and $\mathcal{M}$ are dual to each other.

\begin{lem} \label{lem:conic_duality}
The closure of $\mathcal{N}$ is the dual to the cone of M-matrices $\mathcal{M}$, i.e.
\begin{equation}
\label{conic_duality}
 \overline{\mathcal{N}} \,=\,\bigl\{ \, S\in \mathbb{S} \,\mid\, \langle S, K \rangle \geq 0
\,\,\,\hbox{for all} \,\,\,K \in \mathcal{M} \bigr\}.
\end{equation}
\end{lem}

\begin{proof}
We denote the dual of a convex cone $\mathcal{C}$ by $\mathcal{C}^\vee$. Let $\mathcal{C}_1$, $\mathcal{C}_2$ be two convex cones. Then it is an easy exercise to verify that
\begin{equation}\label{eq:2cones}
(\mathcal{C}_1\cap \mathcal{C}_2)^\vee = \mathcal{C}_1^\vee + \mathcal{C}_2^\vee;	
\end{equation}
here $+$ denotes the Minkowski sum. Note that
$$\mathbb{S}_{\succ 0}^\vee = \mathbb{S}_{\succeq 0} \quad \textrm{ and } \quad \mathcal{H}_{ij}^\vee = \mathcal{H}_{ij}.$$
This completes the proof, since $\mathcal{M}=\mathbb{S}_{\succ 0}\cap_{i< j} \mathcal{H}_{ij}$ and (\ref{eq:2cones}) can be applied inductively to any finite collection of convex cones.
\end{proof}

Using the cones $\mathcal{M}$ and $\mathcal{N}$ we now determine conditions for existence of the MLE in Gaussian $\mtp$ models and give a characterization of the MLE. We say that the MLE does not exist if the likelihood does not attain the global maximum.

\begin{prop}\label{prop:mtp2GaussGraphModel}
Consider a Gaussian $\mtp$ model. Then the MLE $\hat{\Sigma}$ (and $\hat{K}$) exists for a given sample covariance matrix $S$ on $V$ if and only if 
$S\in\mathcal{N}$. It is then equal to the unique element $\hat\Sigma\succ 0$ that satisfies the following system of equations and inequalities
\begin{eqnarray}  
(\hat\Sigma^{-1})_{ij}&\leq&0\; \mbox{ for all $i\neq j$},\label{eq:primal1}\\
\hat\Sigma_{ii}-S_{ii}&=&0\; \mbox{ for all $i\in V$},\label{eq:dual1}\\ 
(\hat\Sigma_{ij}-S_{ij})&\geq & 0\;\mbox{ for all $i\neq j$}, \label{eq:dual2}\\
(\hat\Sigma_{ij}-S_{ij})(\hat\Sigma^{-1})_{ij}&=& 0\;\mbox{ for all $i\neq j$}, \label{eq:slackness}
\end{eqnarray}
\end{prop}

\begin{proof}  It is straight-forward to compute the dual optimization problem and the KKT conditions. In particular, in~\cite{slawski2015estimation} it was shown that the dual optimization problem to (\ref{ML_mtp2}) is given by
\begin{equation}
\label{ML_mtp2_dual}
\begin{aligned}
& \underset{\Sigma\succeq 0}{\text{minimize}}
& & -\log\det(\Sigma)-p  \\
& \text{subject to}
&& \Sigma_{ii} = S_{ii}, \; \mbox{ for all $i\in V$},\\
&&& \Sigma_{ij} \geq S_{ij}, \; \mbox{ for all $i\neq j$}.\\
\end{aligned}
\end{equation}

Note that the identity matrix is a strictly feasible point for (\ref{ML_mtp2}). Hence, the MLE does not exist if and only if the likelihood is unbounded. Since by Slater's constraint qualification strong duality holds for the optimization problems (\ref{ML_mtp2}) and (\ref{ML_mtp2_dual}), the MLE does not exist if and only if $S\notin\mathcal{N}$. 
\end{proof}
We note that the conditions in Proposition~\ref{prop:mtp2GaussGraphModel} were also derived in \cite{slawski2015estimation}, save for the explicit identification of the dual cone $\mathcal{N}$.

\begin{rem}\label{rem:ggm} Proposition~\ref{prop:mtp2GaussGraphModel} can easily be extended to provide properties for the existence of the MLE and a characterization of the MLE for Gaussian graphical models under $\mtp$. In this case, let $G=(V,E)$ be an undirected graph. Then the primal problem has additional equality constraints, namely $K_{ij}=0$ for all $ij\notin E$, and hence the inequality constraints in the dual problem are restricted to the entries in $E$, i.e., $\Sigma_{ij} \geq S_{ij}$ for all $ij\in E$. Note that if the MLE of $\Sigma$ based on $S$ exists in the Gaussian graphical model over $G$, it also exists in the Gaussian graphical model over $G$ under $\mtp$, since without the $\mtp$ constraint the MLE needs to satisfy $\hat{\Sigma}_{ij} =S_{ij}$ for all $ij\in E$. 
\qed
\end{rem}

We define the \emph{maximum likelihood graph} (ML graph) $\hat{G}$ to be the graph determined by $\hat K$, i.e.\ $\hat G=G(\hat K)$, where $\hat K=\hat\Sigma^{-1}$ is the MLE of $K$ under $\mtp$. We then have the following important corollary of Proposition~\ref{prop:mtp2GaussGraphModel}.
\begin{cor}\label{cor:MLgraphmle}Consider the Gaussian graphical model determined by $K_{ij}=0$ for $ij\not\in E(\hat G)$, where $\hat{G}$ is the ML graph under $\mtp$. Let $\bar\Sigma$ be the MLE of $\Sigma$ under that Gaussian graphical model (without the $\mtp$ constraint). Then $\hat\Sigma=\bar\Sigma$.
\end{cor}
\begin{proof}The MLE of $\Sigma$ under the Gaussian graphical model with graph $\hat G$ is the unique element $\bar\Sigma\succ 0$ satisfying the following system of equations: 
\begin{eqnarray*}  
\bar\Sigma_{ii}-S_{ii}&=&0\; \mbox{ for all $i\in V$},\\
\bar\Sigma_{ij}-S_{ij}&=& 0\;\mbox{ for all $ij\in E(\hat G)$},\\
(\bar\Sigma^{-1})_{ij}&=& 0\;\mbox{ for all $ij\not\in E(\hat G)$}.
\end{eqnarray*}
Proposition~\ref{prop:mtp2GaussGraphModel} says that also $\hat\Sigma$ satisfies these equations and hence we must have $\bar\Sigma=\hat\Sigma$.
\end{proof}
Note that this corollary highlights the role of the complementary slackness condition (\ref{eq:slackness}) in inducing sparsity of the $\mtp$ solution.

We emphasize that the MLE under $\mtp$ is equivariant w.r.t.\ changes of scale so that without loss of generality we can assume that the sample covariance is normalized, i.e.\ $S_{ii}=1$ or, equivalently, $S=R$, where $R$ is the correlation matrix. For certain of the subsequent developments this represents a convenient simplification.
\begin{lem}\label{lem:equivariance}
	Let $S$ be the sample covariance matrix, $R$ the corresponding sample correlation matrix. Denote by $\hat \Sigma^S$ and $\hat\Sigma^R$ the MLE in Proposition~\ref{prop:mtp2GaussGraphModel} based on $S$ and $R$, respectively. Then 
	$$
	\hat\Sigma^S_{ij}\;\;=\;\;\sqrt{S_{ii}S_{jj}}\,\hat\Sigma^R_{ij}\qquad\mbox{for all }i,j\in V.
	$$ 
\end{lem}
\begin{proof}
Denote by $D$ a diagonal matrix such that $D_{ii}=\sqrt{S_{ii}}$ and $S=DRD$. The likelihood function based on $S$ is 
$$\log\det K-{\rm tr}(SK)\;\;=\;\;\log\det K-{\rm tr}(RDKD).$$
If $K'=DKD$, this can be rewritten as $\log \det K'-{\rm tr}(RK')-\sum_i\log S_{ii}$. Therefore, if $\hat K^R$ is the maximizer of $\log \det K-{\rm tr}(RK)$ under the $\mtp$ constraints, then $D^{-1}\hat K^RD^{-1}$ is also an M-matrix and the maximizer of $\log\det K-{\rm tr}(SK)$.
\end{proof}

We end this section by providing simple coordinate descent algorithms for ML estimation under $\mtp$. Although interior point methods run in polynomial time, for very large Gaussian graphical models it is usually more practical to apply coordinate descent algorithms. In Algorithms~\ref{algthm:IPS_opt_K} and~\ref{algthm:IPS_opt_Z} we describe two methods for computing the MLE that only use optimization problems of size $2\times 2$ which have a simple and explicit solution, and iteratively update the entries of $K$, respectively of $\Sigma$. Algorithms~\ref{algthm:IPS_opt_K} and~\ref{algthm:IPS_opt_Z} are inspired by the corresponding algorithms for Gaussian graphical models; see, for example,  \cite{dempster1972covariance,SK1986,WS1977}. Slawski and Hein~\cite{slawski2015estimation} also provide a coordinate descent algorithm for estimating covariance matrices under $\mtp$. However, their method updates one column/row of $\Sigma$ at a time. \del 

We first analyze Algorithm~\ref{algthm:IPS_opt_K}. Let $A=\{u,v\}$ and $B=V\setminus A$. Then note that the objective function can be written in terms of the $2\times 2$ Schur complement $K' = K_{AA}-K_{AB}K_{BB}^{-1}K_{BA}$, since up to an additive constant
$$\log\det K - \textrm{trace}(KS) = \log\det K' -\textrm{trace}(K'S_{AA}).$$
Defining $L:= K_{AB}K_{BB}^{-1}K_{BA}$, then the optimization problem in step (2) of Algorithm~\ref{algthm:IPS_opt_K} is equivalent to
\begin{equation*}
\begin{aligned}
& \underset{K'\succeq 0}{\text{maximize}}
& & \log\det(K') - \textrm{trace}(K'S_{AA})\\
& \text{subject to}
&& K'_{12} +L_{12} \leq 0.\\
\end{aligned}
\end{equation*}
The unconstrained optimum to this problem is given by $K'=S_{AA}^{-1}$ and is attained if and only if $(S_{AA}^{-1})_{12} + L_{12}\leq 0$, or equivalently, if and only if
$$L_{12} \leq \frac{S_{uv}}{S_{uu}S_{vv}-S_{uv}^2}.$$ 
Otherwise the KKT conditions give that $K'_{12}=-L_{12}$. 

Maximizing over the remaining two entries of $K'$ leads to a quadratic equation, which has one feasible solution
\begin{equation}\label{eq:updK}
K'_{11}=\frac{1+\sqrt{1+4S_{uu}S_{vv}L_{12}^2}}{2S_{uu}},\quad K'_{22}=\frac{1+\sqrt{1+4S_{uu}S_{vv}L_{12}^2}}{2S_{vv}},\quad K'_{12}=-L_{12}.	
\end{equation}
Then the solution to the optimization problem in step (2) is given by $K_{AA} = K'+L$. 

Dual to this algorithm, one can define an algorithm that iteratively updates the off-diagonal entries of $\Sigma$ by maximizing the log-likelihood in direction $\Sigma_{uv}$ and keeping all other entries fixed. This procedure is shown in Algorithm~\ref{algthm:IPS_opt_Z}. If $p>n$, $S$ is not positive definite; in this case we use as starting point the single linkage matrix $Z$ that is defined later in (\ref{eq:Z}). 

Similarly as for Algorithm~\ref{algthm:IPS_opt_K}, the solution to the optimization problem in step (2) can be given in closed-form. Defining $A=\{u,v\}$, $B=V\setminus A$ and $L = \Sigma_{AB}\Sigma_{BB}^{-1}\Sigma_{BA}$, then analogously as in the derivation above, one can show that the solution to the optimization problem in step (2) of Algorithm~\ref{algthm:IPS_opt_Z} is given by
\begin{equation}\label{eq:Sigupdate}
\Sigma_{uv}=\max\{S_{uv}, L_{12}\}.	
\end{equation}

\begin{algorithm}[!t]
\caption{Coordinate descent on $K$.}
\label{algthm:IPS_opt_K}
\begin{algorithmic}
\begin{STATE}
\vspace{0.2cm}
{\bf Input:\;\;\;}
Sample covariance matrix $S$, and precision $\epsilon$.

{\bf Output:} 
MLE $\hat{K}\in \mathcal{M}$.

\vspace{0.2cm}

\begin{enumerate}[1.]

\item Let $K^0 := K^1 := (\textrm{diag}(S))^{-1}$.
\item Cycle through entries $u\neq v$ and solve the following optimization problem:
\begin{equation*}
\begin{aligned}
& \underset{K\succeq 0}{\text{maximize}}
& & \log\det(K) - \textrm{trace}(KS)  \nonumber\\
& \text{subject to}
&& K_{uv} \leq 0,\nonumber\\
&&& K_{ij} = K^1_{ij} \;\textrm{ for all } ij\in (V\times V)\setminus\{uu,vv,uv\},\\  
\end{aligned}
\end{equation*}
and update $K^1 = K$.
\item If $|\!|K^0-K^1|\!|_1 <\epsilon$, set $\hat{K}=K^1$. Otherwise, set $K^0=K^1$ and return to 2.
\end{enumerate}
\end{STATE}
\end{algorithmic}
\end{algorithm}

\begin{algorithm}[!b]
\caption{Coordinate descent on $\Sigma$.}
\label{algthm:IPS_opt_Z}
\begin{algorithmic}
\begin{STATE}
\vspace{0.2cm}
{\bf Input:\;\;\;}
Sample covariance matrix $S\succ 0$, and precision $\epsilon$.

{\bf Output:} 
MLE $\hat{\Sigma}$ with $\hat{\Sigma}^{-1}\in \mathcal{M}$.

\vspace{0.2cm}

\begin{enumerate}[1.]

\item Let $\Sigma^0 := \Sigma^1 := S$
\item Cycle through entries $u\neq v$ and solve the following optimization problem:
\begin{equation*}
\begin{aligned}
& \underset{\Sigma\succeq 0}{\text{maximize}}
& & \log\det(\Sigma)  \\
& \text{subject to}
&& \Sigma_{uv} \geq S_{uv},\\
&&& \Sigma_{ij} = \Sigma^1_{ij} \;\textrm{ for all } ij\in (V\times V)\setminus\{uv\}.
\end{aligned}
\end{equation*}
and update $\Sigma^1 = \Sigma$.
\item If $|\!|\Sigma^0-\Sigma^1|\!|_1 <\epsilon$, set $\hat{\Sigma}=\Sigma^1$. Otherwise, set $\Sigma^0=\Sigma^1$ and return to 2.
\end{enumerate}

\end{STATE}
\end{algorithmic}
\end{algorithm}

We end by proving that Algorithms~\ref{algthm:IPS_opt_K} and~\ref{algthm:IPS_opt_Z} indeed converge to the MLE. We here assume that $n\geq 2$ to guarantee existence of the MLE. Note that the suggested starting points for both algorithms can be modified. 

\begin{prop}
\label{prop_IPS}
Algorithms~\ref{algthm:IPS_opt_K} and~\ref{algthm:IPS_opt_Z} converge to the MLE $\hat{K}=\hat{\Sigma}^{-1}\in\mathcal{M}$. 
\end{prop}

\begin{proof}
The convergence to the MLE is immediate for Algorithm~\ref{algthm:IPS_opt_Z} because it is a coordinate descent method applied to a smooth and strictly concave function; see, e.g., \cite{tseng1992}. For Algorithm~\ref{algthm:IPS_opt_K} we use the fact that it is an example of iterative partial maximization. To prove convergence to the MLE we we will show that the assumptions of Proposition A.3 in \cite{lau96} hold. The log-likelihood function that we are trying to maximize is strictly concave and so the maximum is unique. Clearly, $K$ is the maximum if and only if it is a fixed point of each update. It only remains to show that updates depend continuously on the previous value. For a given $S$ fix $K$ and consider a sequence of points $K_n$ converging to $K$. Denote by $\tilde K$ and $\tilde K_n$ the corresponding one-step updates. We want to show that $\tilde K_n$ also converges to $\tilde K$. As above, let $A=\{u,v\}$, $B=V\setminus A$, $K' = K_{AA}-K_{AB}K_{BB}^{-1}K_{BA}$ and $L = K_{AB}K_{BB}^{-1}K_{BA}$. Outside of the block $\tilde K_{AA}$ this convergence is trivial; so we focus only on the three entries in $\tilde K_{AA}$. The function $L_{12}\mapsto (K'_{11},K_{22}',K_{12}')$ is continuous if and only if each coordinate is. It is clear that these functions are continuous if $L_{12}\neq \frac{S_{uv}}{S_{uu}S_{vv}-S_{uv}^2}$. It remains to show that if $L_{12}=\frac{S_{uv}}{S_{uu}S_{vv}-S_{uv}^2}$ the update in (\ref{eq:updK}) gives $K'=S_{AA}^{-1}$, which can be easily checked. 
\end{proof}

\section{Ultrametric matrices and inverse M-matrices}\label{sec:ultrametric}

In this section we exploit the link to ultrametrics in order to construct an explicit primal and dual feasible point of the maximum likelihood estimation problem. 

A nonnegative symmetric matrix $U$ is said to be \emph{ultrametric} if
\begin{itemize}
	\item[(i)] $U_{ii}\geq U_{ij}$ for all $i,j\in V$,
	\item[(ii)] $U_{ij}\geq \min\{U_{ik},U_{jk}\}$ for all  $i,j,k\in V$.
\end{itemize}
We say that a symmetric matrix is an \emph{inverse M-matrix} if its inverse is an M-matrix. The connection between ultrametrics and  M-matrices is established by the following result; see \cite[Theorem~3.5]{dellacherie2014inverse}. 
\begin{thm}\label{th:PDU}
	Let $U$ be an ultrametric matrix with strictly positive entries on the diagonal. Then $U$ is nonsingular if and only if no two rows are equal. Moreover, if $U$ is nonsingular then $U$  is an inverse M-matrix.  
\end{thm} 

The main reason why ultrametric matrices are relevant here is the following construction, which is similar to constructions used in in phylogenetics \cite[Section 7.2]{semple2003phylogenetics} and single linkage clustering \cite{gower1969minimum}.

Let $R$ be a symmetric $p\times p$ positive semidefinite matrix such that $R_{ii}=1$ for all $i\in V$. 
Consider the weighted graph $G^+= G^+(R)$ over $V$ with an edge between $i$ and $j$ whenever $R_{ij}$ is positive and assign to each edge the corresponding positive weight $R_{ij}$. Note that $G^+$ in general does not have to be connected. Define a $p\times p$ matrix $Z$ by setting $Z_{ii}=1$ for all $i\in V$ and
\begin{equation}\label{eq:Z}
	Z_{ij}\;\;:=\;\;\max_{P} \min_{uv\in P} R_{uv}, 
	\end{equation}
for all $i\neq j$, where the maximum is taken over all paths in $G^+$ between $i$ and $j$ and is set to zero if no such path exists. We call $Z$ the \emph{single-linkage matrix} based on $R$. 

\begin{ex}\label{ex:Z}
	Suppose that 
	$$
	R\;=\;\begin{bmatrix}
		1 &-0.5 &0.5 &0.6 \\
		 -0.5&1 &0.4 &-0.1 \\
		 0.5& 0.4& 1& 0.2\\
		 0.6&-0.1 & 0.2&1 		
	\end{bmatrix}
	$$
	Then $G^+$ and $Z$ are given by	\vspace*{\baselineskip}

\qquad\qquad\begin{minipage}{0.35\textwidth}
\begin{tikzpicture}[auto, node distance=2cm, every loop/.style={},
                    main node/.style={circle,draw,font=\normalsize}]

  \node[main node] (1) {1};
  \node[main node] (2) [below left of=1] {2};
  \node[main node] (3) [below right of=2] {3};
  \node[main node] (4) [below right of=1] {4};

  \path[every node/.style={}]
    (1) edge node [right] {0.6} (4)
        edge [right] node[left] {0.5} (3)
    (2) edge [right] node[left] {0.4} (3)
    (3) edge [right] node[right] {0.2} (4);
\end{tikzpicture}
\end{minipage} 
\begin{minipage}{0.3\textwidth}
$$
	Z\;=\;\begin{bmatrix}
		1 &0.4 &0.5 &0.6 \\
		 0.4&1 &0.4 &0.4 \\
		 0.5& 0.4& 1& 0.5\\
		 0.6&0.4 & 0.5&1 		
	\end{bmatrix}.
	$$
\end{minipage}	
\vspace{0.4cm}

	For example, to get $Z_{12}$ we consider two paths $1-3-2$ and $1-4-3-2$. The minimum of $R_{uv}$ over the first path is $0.4$ and over the second path $0.2$. This gives $Z_{12}=0.4$. 
\qed \end{ex}

Note that in the above example $Z\geq R$, $Z$ is invertible, and $
	Z^{-1}$ is an M-matrix. We now show that this is an example of a more general phenomenon.

\begin{prop}\label{prop:Z}
	Let $R$ be a symmetric $p\times p$ positive semidefinite matrix satisfying $R_{ii}=1$ for all $i\in V$. Then the single-linkage matrix $Z$ based on $R$ is an ultrametric matrix with $Z_{ij}\geq R_{ij}$ for all $i\neq j$. If, in addition,  $R_{ij}<1$ for all $i\neq j$, then $Z$ is nonsingular and therefore an inverse M-matrix.
\end{prop}
\begin{proof}
We first show that $Z$ is an ultrametric matrix. $Z$ is symmetric by definition. Because $R$ is positive semidefinite, $R_{ij}\leq 1$ for all $i,j$ and from (\ref{eq:Z}) it immediately follows that $Z_{ij}\leq 1$ and therefore $Z_{ii}\geq Z_{ij}$ for all $i,j$ as needed. Finally, to prove condition~(ii) in the definition of ultrametric, let $i,j,k\in V$. Suppose first that $i,j,k$ lie in the same connected component of $G^+$. Let $P_1$, $P_2$ be the paths in $G^+$ such that $Z_{ik}=\min_{uv\in P_1} R_{uv}$  and $Z_{jk}=\min_{uv\in P_2} R_{uv}$. Let $P_{12}$ be the path between $i$ and $j$ obtained by concatenating $P_1$ and $P_2$. Then
$$
Z_{ij}\;\;=\;\;\max_{P} \min_{uv\in P} R_{uv}\;\;\geq \;\; \min_{uv\in P_{12}} R_{uv}=\min\{Z_{ik},Z_{jk}\}.
$$ 
Now suppose that $i,j,k$ are not in the same connected component of $G^+$. In that case $0\in \{Z_{ij},Z_{ik},Z_{jk}\}$. Because zero is attained at least twice, again $Z_{ij}\geq \min\{Z_{ik},Z_{jk}\}$. Hence,  $Z$ is an ultrametric matrix. The fact that $Z_{ij}\geq R_{ij}$ for all $i,j$ follows directly by noting that the edge $ij$ forms a path between $i$ and $j$. 
\medskip

Suppose now that $R_{ij}<1$ for all $i\neq j$. In that case also $Z_{ij}<1$ for all $i\neq j$. From this it immediately follows that no two rows of $Z$ can be equal. Indeed, if the $i$-th row is equal to the $j$-th row for some $i\neq j$, then necessarily $Z_{ij}=Z_{ii}=Z_{jj}$, a contradiction.  From Theorem \ref{th:PDU} it then follows that $Z$ is an inverse M-matrix, which completes the proof. 
\end{proof}
As a direct consequence we obtain the following result. 
\begin{prop}\label{prop:ZZ}
Let $S$ be a symmetric positive semidefinite matrix with strictly positive entries on the diagonal and such that $S_{ij}<\sqrt{S_{ii}S_{jj}}$ for all $i\neq j$. Then there exists an inverse M-matrix $Z$ such that $Z\geq S$ and $Z_{ii}=S_{ii}$ for all $i\in V$. 	
\end{prop}
\begin{proof}
	Apply Proposition \ref{prop:Z} to the normalized version $R$ of $S$, with entries $R_{ij}:=S_{ij}/\sqrt{S_{ii}S_{jj}}$. Because $R_{ij}<1$ for all $i\neq j$, the corresponding single-linkage matrix $Z'$ is ultrametric with $Z'\geq R$ and $Z'$ is an inverse M-matrix. Define $Z$ by $Z_{ij}=\sqrt{S_{ii}S_{jj}}Z'$. Then $Z\geq S$ and $Z_{ii}=S_{ii}$ for all $i\in V$. Moreover, $Z$ is an inverse M-matrix because $Z'$ is.
\end{proof}

Proposition \ref{prop:ZZ} is very important for our considerations. A basic application is an elegant alternative proof of the main result of \cite{slawski2015estimation}, which says that the MLE under $\mtp$ exists with probability one as long as $n\geq 2$. This is in high contrast with the existence of the MLE in Gaussian graphical models without additional constraints; see~\cite{uhler2010}.
\begin{thm}[Slawski and Hein~\cite{slawski2015estimation}]\label{th:SH}
Consider a Gaussian $\mtp$ model and let $S$ be the sample covariance matrix. If $S_{ij}<\sqrt{S_{ii}S_{jj}}$ for all $i\neq j$ then the MLE $\hat\Sigma$ (and $\hat K$) exists and it is unique. In particular, if the number $n$ of observations satisfies $n\geq 2$, then the MLE exists with probability $1$. 
\end{thm}
\begin{proof}
	The sample covariance matrix is a positive semidefinite matrix with strictly positive diagonal entries. We can apply  Proposition  \ref{prop:ZZ} to obtain an inverse M-matrix $Z$ that satisfies $Z\geq S$ and $Z_{ii}=S_{ii}$ for all $i$. It follows that $Z$ satisfies primal feasibility (\ref{eq:primal1}) and dual feasibility (\ref{eq:dual1}) and (\ref{eq:dual2}). By Proposition \ref{prop:mtp2GaussGraphModel} the MLE exists and it is unique by convexity of the problem.
\end{proof}
\add{\begin{rem} Combining this result with Corollary \ref{cor:MLgraphmle} we note that the cliques of  $\hat G$ can at most be of size $n$. In this way the sparsity of $\hat G$ automatically adjusts to the sample size.\end{rem}}

The matrix $Z$ can be computed efficiently\footnote{In our computations we use the single-linkage clustering method in \textsc{R}.}. To see that, note first that in Example~\ref{ex:Z} we could first consider the chain $T$ of the form $2-3-1-4$, which is the maximal weight spanning forest of $G^+$ and then construct $Z$ by
$$
Z_{ij}\;=\;\min_{uv=\overline{ij}} R_{uv},
$$ 
where $\overline{ij}$ denotes the unique path between $i$ and $j$ in $T$. For example $Z_{12}=0.4$, which corresponds to the minimal weight on the path $2-3-1$. This is a general phenomenon.

Suppose again that $R$ is a symmetric $p\times p$ positive semidefinite matrix satisfying $R_{ii}=1$ for all $i\in V$. Let $\mwsf(R)$ be the set of all minimal maximum weight spanning forests of $R$. Note that all edge weights of any such forest $F\in \mwsf(R)$ must be positive; hence we must have $F\subseteq G^+$. Also, if $R$ is an empirical correlation matrix, then $\mwsf(R)$ will be a singleton with probability one and in such cases we shall mostly speak of \emph{the} MWSF. 

\begin{prop}\label{prop:ZZZ}
	The single-linkage matrix $Z$ as defined in (\ref{eq:Z}) is block diagonal with blocks corresponding to the connected components of any $F\in \mwsf(R)$. Within each block all elements are strictly positive and given by
	$$
	Z_{ij}\;\;=\;\;\min_{uv\in \overline{ij}} R_{uv}, 
	$$
	where $\overline{ij}$ is the unique path between $i$ and $j$ in a maximal weight spanning tree of $R$. In particular, $Z_{ij}=R_{ij}$ for all edges of $\mwsf(R)$. 
\end{prop}

\begin{proof}
First suppose that $i,j\in V$ lie in two different components of $F\in \mwsf(R)$. This means that there is no path between $i$ and $j$ in $G^+$ and so, by definition, $Z_{ij}=0$. Because $Z_{ij}>0$ if $i,j$ lie in the same component of $F$, $Z$ is block diagonal with blocks corresponding to connected components of $\mwsf(R)$. 

The rest of the proof is an adaptation of a proof of a related result \cite[Proposition 7.2.10]{semple2003phylogenetics}. Suppose that  $i,j\in V$ lie in the same connected component of $F$ and denote the tree in $F$ corresponding to this component by $T$. By definition $Z_{ij}\geq 	\min_{uv\in \overline{ij}} R_{uv}$. Suppose that $Z_{ij}>	\min_{uv\in \overline{ij}} R_{uv}$. We obtain the contradiction by showing that under this assumption $T$ cannot be a maximum weight spanning tree of the corresponding connected component of $G^+$. Let $kl$ be a minimum weight edge in the unique path between $i$ and $j$ in $T$. Since $Z_{ij}>R_{kl}$, there exists a path $P$ in $G^+$ between $i$ and $j$ such that $R_{uv}>R_{kl}$ for every $uv$ in $P$. Now deleting $kl$ from $T$ partitions the corresponding connected component of $G^+$ into two sets with $i$ being in one and $j$ being in the other block. Since $P$ connects $i$ and $j$ in $G^+$, there must be an edge $k'l'$ (distinct from $kl$) in $P$ whose end vertices lie in different blocks of this partition. Let $T'$ be the spanning tree obtained from $T$ by deleting $kl$ and adding $k'l'$. Since $R_{k'l'}>R_{kl}$, the total weight of $T'$ is greater than $T$, which is a  contradiction. We conclude that $Z_{ij}=\min_{uv\in \overline{ij}} R_{uv}$ for all $i,j$ in the same connected component of~$G^+$. 
\end{proof}

To conclude this section, we note that the starting point $\Sigma^0$ of Algorithm~\ref{algthm:IPS_opt_Z} is arbitrary as long as $\Sigma^0\succ 0$ and $\Sigma^0\geq S$. The single-linkage matrix \add{$Z$} constitutes another generic choice when $S=R$ is used as input. This is a particularly desirable starting point, since it can also be used when $p>n$, in which case $R\notin\mathbb{S}_{\succ 0}$ and hence not feasible.

\section{The maximum likelihood graph}
\label{sec:MLgraph}

Fitting a Gaussian model with $\mtp$ constraints tends to induce sparsity in the maximum likelihood estimate $\hat K$. In this section, we analyze the sparsity pattern that arises in this way. We assume again without loss of generality that $S=R$ is a sample correlation matrix so that $R_{ii}=1$ for all $i$ and $R_{ij}<1$ for all $i\neq j$. Consider again the weighted graph $G^+=G^+(R)$. 
We begin this section with a basic lemma that reduces our analysis to the case where the graph $G^+$ is connected.  
\begin{lem}
The MLE $\hat\Sigma$ under $\mtp$ is a block diagonal matrix with strictly positive entries in each block. The blocks correspond precisely to trees in $\mwsf(R)$.
\end{lem}
\begin{proof}
Firstly, since $\hat\Sigma$ is an inverse M-matrix, it is block diagonal with strictly positive entries in each block; see, e.g.,  Theorem 4.8 in \cite{johnson2011inverse}. 
We will show that each block of $\hat\Sigma$ corresponds precisely to a tree in $\mwsf(R)$. 

Denote the vertex sets for a forest  $F\in \mwsf(R)$ as $T_1,\ldots, T_k$ and the blocks of $\hat\Sigma$ as $B_1,\ldots, B_l$. Firstly, for any $T_i$ there must be a $j$ so that $T_i\subseteq B_j$; this is true since all entries in $R$ along the edges of $T_i$ are positive and thus $\hat\Sigma\geq R>0$.  Thus the block partitioning corresponding to the trees is necessarily finer than that of $\hat\Sigma$.

On the other hand, suppose that two different trees $T_i$ and $T_j$ in $F$ are in the same block of $\hat\Sigma$ so that $\hat\Sigma_{uv}>0$ for all $u\in T_i$ and $v\in T_j$. Then, as we must have $R_{uv}\leq 0$, also necessarily $\hat\Sigma_{uv}-R_{uv}>0$. Complementary slackness~(\ref{eq:slackness}) now implies that $\hat K_{uv}=0$ for all $u\in T_i$ and $v\in T_j$, and hence $\hat K$ is block-diagonal with blocks corresponding to the trees in $F$. Since $\hat\Sigma=\hat K^{-1}$, we also get $\hat\Sigma_{uv}=0$ which contradicts that $u$ and $v$ are in the same block of $\hat\Sigma$.
\end{proof}
This result shows that, without loss of generality, we can always assume that $G^+$ is connected and then $\mwsf(R)=\mwst(R)$ consists of trees only. If there are more than one connected component, we simply compute the MLE for each component separately and combine them together in block diagonal form. Hence, from now on we always assume that all forests in $\mwsf(R)$ are just trees.

\del{}

\subsection{An upper bound on the ML graph}

{\del}
In the following, we provide a simple procedure for identifying an upper bound for $\hat G$. 
This procedure relies on the estimation of the standard Gaussian graphical model over the tree $\mwsf(R)$. The MLE under this assumption, denoted by $\tilde\Sigma$, can be computed efficiently and it satisfies
$$
\tilde\Sigma_{ij}\;\;=\;\;\prod_{uv\in \overline{ij}}R_{uv}.
$$
where $\overline{ij}$ denotes the unique path between $i$ and $j$ in $\mwsf(R)$; see, for example, \cite[Section 8.2]{LTbook}. 

To provide an upper bound on $\hat G$, we 
will make use of a connection to so-called path product matrices: A non-negative matrix $R$ is a \emph{path product matrix} if  for any $i,j\in V$, $k\in \N$, and $1\leq i_1,\ldots,i_k\leq p$ $$R_{ij}\;\;\geq \;\;R_{i i_1}R_{i_1 i_2}\cdots R_{i_k j}.$$
If in addition the inequality is strict for $i= j$, we say that $R$ is a \emph{strict path product matrix}. We note the following:
\begin{thm}[Theorem 3.1, \cite{johnson1999path}]\label{thm:PP}
	Every inverse M-matrix is a strict path product matrix.
\end{thm} 

We are now able to provide an upper bound for the ML graph $\hat G$.

\begin{prop}\label{prop:PP}
The pair $ij$ forms an edge in the ML graph only if
\del
\add{
$$
R_{ij}\;\;\geq\;\; \prod_{uv\in P}R_{uv}
$$ for any path in $P$ in $G^+$.	}
In particular, $R_{ij}\leq 0$ implies that $ij$ is not an edge of the ML graph.
\end{prop}
\begin{proof}
Because $\hat\Sigma$ is an inverse M-matrix it is necessarily a path product matrix by Theorem \ref{thm:PP}. In particular, for all $i,j$
$$
\hat\Sigma_{ij}\;\;\geq\;\; \prod_{uv\in P}\hat\Sigma_{uv}.
$$	for any path $P$ in $G^+$. 
By Proposition~\ref{prop:mtp2GaussGraphModel}, we also have $\hat\Sigma_{uv}\geq R_{uv}$. Thus, if $ij\in \hat G$ 
\[R_{ij}=\hat\Sigma_{ij}\geq \prod_{uv\in P}\hat\Sigma_{uv}\geq \prod_{uv\in P}R_{uv}\]
as desired.
\end{proof}
Motivated by this result we define the \emph{excess correlation graph} (EC graph) $\ec(R)$ of $R$ by the condition
\[i\sim j \;\;\iff\;\; R_{ij}\geq \prod_{uv\in \overline{ij}}R_{uv}.\]
Thus the EC graph has edges $ij$ whenever the observed correlation between $i$ and $j$ is in excess of or equal to what is explained by the spanning forest; by construction, \del
\[G(\hat K)\;\;\subseteq\;\; \ec(R).\]
The inclusion \del
\add{is}
typically strict. For example, if $R$ is an inverse M-matrix, then $\ec(R)$ is the complete graph, whereas $G(\hat K)$ can be arbitrary; this follows from~\cite[Proposition~6.3]{MTP2Markov2015}.

\subsection{Some exact results on the ML graph}
\del

Next, we analyze generalization of trees known as block graphs, where edges are replaced by cliques, and give a condition under which the maximum likelihood estimator admits a simple closed-form solution. More formally, $G$ is a \emph{block graph} if it is a chordal graph with only singleton separators. It is natural to study block graphs, since viewing the MLE $\hat{\Sigma}$ as a completion of $S$, block graphs play the same role for inverse M-matrices as chordal graphs play for Gaussian graphical models, see for example \cite{johnson:smith:96} and Corollary~7.3 of \cite{MTP2Markov2015}.

We first define a matrix $W=W(R)$~by
\begin{equation}\label{eq:W}
	W_{ij}\;\;:=\;\;\max_{P} \prod_{uv\in P} R_{uv}, 	
\end{equation}
where, like in (\ref{eq:Z}), the maximum is taken over all paths in $G^+$ between $i$ and $j$ and is set to zero if no such path exists. Transforming $D_{ij}=-\log R_{ij}$ gives a distance based interpretation, in which $W_{ij}$ is related to the shortest distance between $i$ and $j$ in $G^+$ with edge lengths given by $D_{uv}$. We also have the following simple lemma.
\begin{lem}\label{lem:PP}
	The matrix $W$ is a path product matrix. Further, $R$ is a path product matrix if and only if $ W(R)=R$.
\end{lem}
\begin{proof}This is immediate from the definition of $W$.\end{proof}

It is easy to show that $Z\geq W\geq R$ and that $W$ is always equal to the MLE $\hat\Sigma$ in the case when $p\leq 3$. For general $p$ we do not know conditions on $R$ that assure that $W$ is an inverse M-matrix, or the MLE. Indeed, Example 3.4 in \cite{johnson1999path} gives a strict path product correlation matrix $R$, and thus $W=R$, which is not an inverse M-matrix, and thus $W\neq \hat \Sigma$. We note that $W=\hat\Sigma$ for the \texttt{carcass} data discussed in Example~\ref{ex:intro} and, as we shall see in the following, it reflects that in this example, the ML graph is   a block graph.  

Let $G_R(W)$ be the graph having edges $ij$ exactly when $R_{ij}=W_{ij}$ and no edges otherwise. We then obtain the following result. 
\begin{prop}\label{prop:blockest}
	If $G_R(W)$ is a block graph and blocks of $W$ corresponding to cliques are inverse M-matrices, then $\hat\Sigma=W$ and $\hat G\subseteq G_R(W)$.
\end{prop}
\begin{proof} Note first that if $\hat \Sigma=W$, the KKT conditions (\ref{eq:slackness}) imply that $\hat G\subseteq G_R(W)$. Let $\tilde\Sigma$ denote the maximum likelihood estimate of $\Sigma$ under the Gaussian graphical model with graph $G_R(W)$. Then, since $G_R(W)$ is a block graph, it follows from~\cite[equation (5.46) on page 145]{lau96} that $\tilde \Sigma$ is an inverse M-matrix  which coincides with $W$ and $R$ on all edges of $G_R(W)$. So from to show that
$\tilde \Sigma=\hat \Sigma=W$ we just need to argue that 
$\tilde \Sigma=W$.

We proceed by induction on the number $m=|\mathcal{C}|$ of cliques of $G_R(W)$.  
If there is only one clique in $G_R(W)$, we have  $\tilde\Sigma=R$ and $R$ is an inverse M-matrix and hence $\hat\Sigma=R=W$. Assume now that the statement holds for $|\mathcal{C}|\leq m$ and assume $G_R(W)$ has $m+1$ cliques.  Since $G_R(W)$ is a block graph, there is a decomposition $(A,B,S)$ of  $G_R(W)$ into block graphs with at most $m$ cliques and with the separator $S=\{s\}$ being a singleton.  
But for a decomposition of $G_R(W)$ as above we have  from~\cite[equation (5.31) in Proposition 5.6]{lau96} and the inductive assumption that
\[\tilde \Sigma_{A\cup S}=\hat\Sigma_{A\cup S} =W(R_{A\cup S}), \quad \tilde \Sigma_{B\cup S}=\hat\Sigma_{B\cup S} =W(R_{B\cup S}).\] 

Now let $P^*$ be the path in $G^+$ such that
$W_{ij}=\prod_{uv\in P^*} R_{uv}$ for any two vertices $i,j$. We claim that all edges in $P^*$ must be edges of $G_R(W)$.  Otherwise, suppose $P^*$ contains an edge $uv$ which is not an edge in $G_R(W)$; then $W_{uv}>R_{uv}$ and so if we replace the edge $uv$ with the path realizing $W_{uv}$ the product would be strictly increased, which contradicts the optimality of $P^*$. Since $S$ is a singleton separator, this also implies that $P^*$ passes through $S$ whenever it involves vertices from both $A$ and $B$. Suppose that $i,j\in A\cup S$. Then optimality of $P^*$ implies that $P^*$ is contained in $A\cup S$ and so $\tilde\Sigma_{A\cup S}=W(R_{A\cup S})=W_{A\cup S}$ and by the same argument $\tilde\Sigma_{B\cup S}=W_{B\cup S}$. Moreover, if $i\in A$ and $j\in B$ then $W_{ij}=W_{is}W_{sj}$. 
Now the inductive assumption in combination with the expression \cite[page 140]{lau96} yields that
\[\tilde \Sigma_{ij}=\tilde\Sigma_{is}\tilde\Sigma_{sj} =W_{is}W_{sj}=W_{ij} \mbox{ for } i\in A, j\in B\] and thus $\tilde\Sigma=\hat\Sigma=W$ as required. \end{proof}

\begin{rem} We note that with probability one, the slackness constraints in (\ref{eq:slackness}) are not simultaneously active, and hence in Proposition~\ref{prop:blockest} we have almost sure equality between $G_R(W)$ and $\hat G$. Thus we can identify $\hat G$ without first calculating $\hat K$.
\end{rem}
\add{We further have the following corollary.
\begin{cor}\label{cor:blockest}Under the same conditions as in Proposition~\ref{prop:blockest} we have $\mwsf(R)\subseteq \hat G\subseteq G_R(W)$.
\end{cor}
\begin{proof}Consider an edge $ij$ between vertices in different cliques of $G_R(W)$ and assume $S_1=\{s_1\}, S_2=\{s_2\}$ are $(i,j)$-separators with $i\sim s_1$ and $j\sim s_2$. Then, since $\hat G\subseteq G_R(W)$ we have $i\cip j\cd s_1$ and $i\cip j\cd s_2$ according to $\hat \Sigma$ and therefore
\begin{eqnarray*}R_{ij}&\leq& \hat \Sigma_{ij}=\hat \Sigma_{is_1}\hat \Sigma_{js_1}=\hat \Sigma_{is_2}\hat \Sigma_{js_2}\\&=&R_{is_1}\hat \Sigma_{js_1}=\hat\Sigma_{is_1}R_{js_1}<\min\{R_{is_1},R_{js_2}\},\end{eqnarray*}so the edge $ij$ can never be part of a $\mwsf$ because  removing the edge would render either $s_1$ disconnected from $i$ or $s_2$ disconnected from $j$ and then the weight of a $\mwsf$ would increase when replacing $ij$ with  $is_1$ or $js_2$, respectively. This completes the proof. 
\end{proof}
It is not correct in general that $\mwsf(R)\subseteq \hat G$ as demonstrated in the following example; although this has been the case in all non-constructed examples we have considered including the relatively large Example~\ref{ex:factor} below.
\begin{ex}\label{ex:nonmwsf}Below we display an inverse M-matrix $K$
\[K=\begin{pmatrix}1 & - .116&  0& 0& - .433\\  - .116 & 1 & - .097 & - .034 & 0\\ 0& - .097 & 1 & -.149 &- .413 \\ 0& -.034& -.149&1& -.604\\ -.433& 0 &
- .413  &-.604 &1 \end{pmatrix}\]
and the corresponding correlation matrix
\[R=         \begin{pmatrix}1 & .2861 & .5745 & .6242& .7299\\
.2861 &1& 0.2864 & .2696 & .2872\\
.5745 &.2864 & 1 &.7149& .7800\\
0.6242& 0.2696&.7149& 1&.8523\\
.7299& .2872& .7800 &.8523 &1
\end{pmatrix}
\]Here the $\mwsf(R)$ is the star graph with $5$ as its center, but the edge $2\sim 5$ is not in $G(K)$. Note that all the edges in $G^+$ adjacent to $2$ have almost the same weight. We note that we have also calculated $K^{-1}$ using rational arithmetic to ensure the phenomenon cannot be explained by rounding error.
\end{ex}
}

\section{Gaussian signed $\mtp$ distributions}
\label{sec:signed}

In this section we discuss how our results can be generalized to so-called signed $\mtp$ Gaussian distributions, where the distribution is $\mtp$ up to sign swapping. Such distributions were discussed by Karlin and Rinott~\cite{karlin1981signed}. More precisely, a random variable $X$ has a \emph{signed $\mtp$ distribution} if there exists a diagonal matrix $D$ with $D_{ii}=\pm 1$ (called \emph{sign matrix}) such that $DX$ is $\mtp$. The following characterization of signed $\mtp$ Gaussian distributions is a direct consequence of \cite[Theorem 3.1 and Remark 1.3]{karlin1981signed}.
\begin{prop}
A Gaussian random variable $X$ has a signed $\mtp$ distribution if and only if $|X|$ is $\mtp$.  	
\end{prop}

Gaussian graphical models with signed $\mtp$ distributions are called \emph{non-frustrated} in the machine learning community. The following result is implicitly stated in  \cite{malioutov2006walk}.

\begin{thm}\label{th:completeGaussSigned}
		A Gaussian random variable $X$ with concentration matrix $K$ has a signed $\mtp$ distribution if and only if it holds for every cycle $(i_1,\ldots,i_k,i_1)$ in the graph $G(K)$ that
	\begin{equation}\label{eq:cycleIneq}
	(-1)^k K_{i_1 i_2}K_{i_2i_3}\cdots K_{i_ki_1}> 0.
	\end{equation}
\end{thm}
\begin{proof}
The 'only if' direction is easy to check. Note that (\ref{eq:cycleIneq}) can be rephrased by saying that each cycle in the graph with edge weights given by the off-diagonal entries of $-K$ has an even number of negative edges. The 'if' direction can now be recovered from the proof of \cite[Corollary 3]{malioutov2006walk}. \end{proof}

Signed $\mtp$ distributions are relevant, for example, because of their appearance when studying tree models.
\begin{prop}\label{prop:trees}
Every Gaussian graphical model over a tree consists of signed $\mtp$ distributions. The $\mtp$ distributions among those are precisely those without negative entries in the covariance matrix $\Sigma$.
\end{prop}
\begin{proof}
Let $T$ be a tree and $K=\Sigma^{-1}$ be a concentration matrix in the Gaussian graphical model over $T$. Then $G(K)$ is a subgraph of $T$ and in particular it has no cycles. Hence by Theorem~\ref{th:completeGaussSigned} it is signed $\mtp$. The second part of the statement follows from~\cite[Corollary~7.3]{MTP2Markov2015}.
\end{proof}

Because signed $\mtp$ distributions are closed under taking margins, Proposition~\ref{prop:trees} can be further generalized. The following theorem covers, in particular, Examples 4.1-4.5 in \cite{karlin1981signed}.
\begin{thm}\label{th:Ltrees}
	Every distribution on a Gaussian tree model with hidden variables is signed $\mtp$. 
\end{thm}

  Gaussian tree models with hidden variables have many applications, in particular related to modeling evolutionary processes; see, e.g., \cite{willsky2011,ASSZ2014}. As an important submodel they contain the Brownian motion tree model \cite{felsenstein_maximum-likelihood_1973}.  Another example of a Gaussian tree model is the factor analysis model with a single factor; it corresponds to a Gaussian model on a star tree, whose inner node is hidden. The $\mtp$ distributions in this model correspond to the distributions in a Spearman model~\cite{ledermann1940problem,Spearman:1928qy}, where  the hidden factor is interpreted as intelligence. 

\bigskip

Let $R$ be a sample correlation matrix. Maximizing the likelihood over all signed $\mtp$ Gaussian distributions requires determining the sign matrix $D$, with $D_{ii}=\pm 1$, that maximizes the likelihood for all $2^p$ possible matrices $DRD$. A natural heuristic is to choose $D^*$ such that $D_{ii}^*D_{jj}^*R_{ij}\geq 0$ for all edges $ij$ of $\mwsf(|R|)$, where $|R|$ denotes the matrix whose entries are the absolute values of the entries of $R$. We provide conditions under which this procedure indeed leads to the MLE under signed $\mtp$, and we also provide examples showing that this is not true in general. Quite interestingly, balanced graphs again play an important role in this part of the theory. 

First we describe how to obtain a sign swapping matrix $D^*$ such that $D_{ii}^*D_{jj}^*R_{ij}\geq 0$ for all edges $ij$ of $\mwsf(|R|)$. Root $\mwsf(|R|)$ at node $1$, that is, regard $\mwsf(|R|)$ as a directed tree with all edges directed away from $1$. Set $D_{11}^*=1$. Then proceed recursively. For any edge $i\to j$ suppose that $D_{ii}^*$ is known and set $D_{jj}^*:={\rm sgn}(D^*_{ii}R_{ij})$. Note that by construction
\begin{equation}\label{eq:Dstar}
D^*_{ii}:= {\rm sgn}(R_{1i_1}R_{i_1i_2}\cdots R_{i_k i}),	
\end{equation}
where $1\to i_1\to i_2\to \cdots\to i_k\to i$ is the unique path from $1$ to $i$ in $\mwsf(|R|)$. We set $D^*_{ii}=0$ if no such path exists. It is easy to check that the resulting $D^*$ satisfies $D_{ii}^*D_{jj}^*R_{ij}\geq 0$ for all edges $ij$ of $\mwsf(|R|)$.

\begin{prop}\label{prop:3S}Suppose that $R$ is a sample correlation matrix whose graph is balanced, that is, such that for every cycle $(i_1,i_2,\ldots,i_k,i_1)$ in the graph $G(R)$
\begin{equation}\label{eq:3S}
R_{i_1i_2}R_{i_2i_3}\cdots R_{i_{k} i_1}> 0.
\end{equation}
Then the MLE based on $R$ over signed $\mtp$ Gaussian distributions is equal to the MLE based on the sample correlation matrix $D^*RD^*$ over $\mtp$ distributions.
\end{prop}
\begin{proof}
We first show that $D^*RD^*$ has only positive entries. Let $i,j$ be any two nodes and let $1\to i_1\to\cdots \to i_k\to i$ and $1\to j_1\to\cdots\to j_l\to j$ be the paths in ${\rm MWSF}(|R|)$ from $1$ to $i$ and $j$, respectively. By (\ref{eq:Dstar}) we obtain
$$
{\rm sgn}(D_{ii}^*D_{jj}^*R_{ij})\;=\;{\rm sgn}(R_{1i_1}\cdots R_{i_k i} R_{ij} R_{j j_l}\cdots R_{j_1 1}),
$$
which is positive by (\ref{eq:3S}). This shows that without loss of generality we can assume that all entries of $R$ are nonnegative and hence that $D^*$ is the identity matrix $\mathbb I_p$. We now  show that the likelihood over $\mtp$ distributions given the sample correlation matrix $DRD$ is maximized by $D=\mathbb{I}_p$. This is because $(D_{ii}D_{jj}-1)\leq 0$ and $R_{ij}K_{ij}\leq 0$, and hence
$$\ell(K;R)-\ell(K;DRD)={\rm tr}(DRDK)-{\rm tr}(RK)=\sum_{i,j}(D_{ii}D_{jj}-1)R_{ij}K_{ij}\geq 0,$$
which completes the proof.
\end{proof} \add{Note that any spanning tree $T$ of $G^+(|R|)$ would suffice to identify the sign switches as above.}

Proposition \ref{prop:3S} provides a sufficient condition for $D^*$ to be the optimal sign-switching matrix; i.e., it provides a sufficient condition such that for every $K\in \mathbb S_{\succ 0}$ and every sign matrix $D$ it holds that 
$$
\ell(K;D^*RD^*)\;\;\geq\;\; \ell(K;DRD).
$$
As a consequence of Proposition~\ref{prop:3S} \del
we obtain the following result for the case when the sample size is 2.

\begin{cor}If the sample correlation matrix $R$ is based on $n=2$ observations,  then the MLE over signed $\mtp$ Gaussian distributions given $R$ is equal to the MLE over $\mtp$ Gaussian distributions given the modified sample correlation matrix $D^*RD^*$. \del 
\end{cor}

Note that the case $n=2$ is special and Proposition~\ref{prop:3S} does not extend to arbitrary sample correlation matrices. In the following, we give a simple counterexample.

\begin{ex}\label{ex:notsign}
Suppose that the sample correlation matrix is
$$R = \begin{bmatrix} 
1 & 0.3 & 0.11 & 0.3 \\
0.3 & 1 & -0.1 & -0.1\\
0.11 & -0.1 & 1 & -0.1\\
0.3 & -0.1 & -0.1 & 1
\end{bmatrix}.$$
Then  $\mwsf(|R|)$ is given by the star graph with edges $1-2$, $1-3$, $1-4$. Since $R$ is positive on these entries, $D^*=\mathbb{I}_p$. But one can check that the corresponding MLE has a lower likelihood than the MLE after changing the sign of the third variable.  

The intuition is the following. The log-likelihood based on $R$ is up to an additive constant given by
\begin{equation*}
\begin{aligned}
& \underset{\Sigma}{\text{minimize}}
& & -\log\det(\Sigma) \\
& \text{subject to}
&& \Sigma_{11} = \Sigma_{22} = \Sigma_{33} = \Sigma_{44} =1,\\
&&& \Sigma_{12} \geq R_{12},\; \Sigma_{13} \geq R_{13},\; \Sigma_{14} \geq R_{14}, \\
&&& \Sigma_{23} \geq 0,\; \Sigma_{24} \geq 0,\; \Sigma_{34} \geq 0, \\
&&& \Sigma \succeq 0.&
\end{aligned}
\end{equation*}
By changing the sign of the third variable, we replace the constraint $1-3$ by two constraints $2-3$ and $3-4$. The resulting optimization problem is
\begin{equation*}
\begin{aligned}
& \underset{\Sigma}{\text{minimize}}
& & -\log\det(\Sigma) \\
& \text{subject to}
&& \Sigma_{11} = \Sigma_{22} = \Sigma_{33} = \Sigma_{44} =1,\\
&&& \Sigma_{12} \geq R_{12},\; \Sigma_{14} \geq R_{14},\; \Sigma_{23} \geq -R_{23},\; \Sigma_{34} \geq -R_{34},\; \\
&&& \Sigma_{13} \geq 0,\; \Sigma_{24} \geq 0, \\
&&& \Sigma \succeq 0.&
\end{aligned}
\end{equation*}
Note that $R_{13}$ is only slightly larger than $-R_{23}$ and $-R_{24}$. Hence, in essence we are increasing the number of constraints by one, which explains the decrease of the log-likelihood value. 
\qed
\end{ex}

We conclude this paper by illustrating how our results can be applied to factor analysis in psychometrics.

\begin{ex}\label{ex:factor}
Single factor models are routinely used to study the personalities in psychometrics. Consider the following example from~\cite{malle1995puzzle}\footnote{\label{footnote_url}We downloaded the data from \url{http://web.stanford.edu/class/psych253/tutorials/FactorAnalysis.html}.}: 240 individuals were asked to rate themselves on the scale 1-9 with respect to 32 different personality traits. The resulting correlation matrix is shown in Figure \ref{fig:FAcorr}.  
\begin{figure}[t]
\includegraphics[scale=.5]{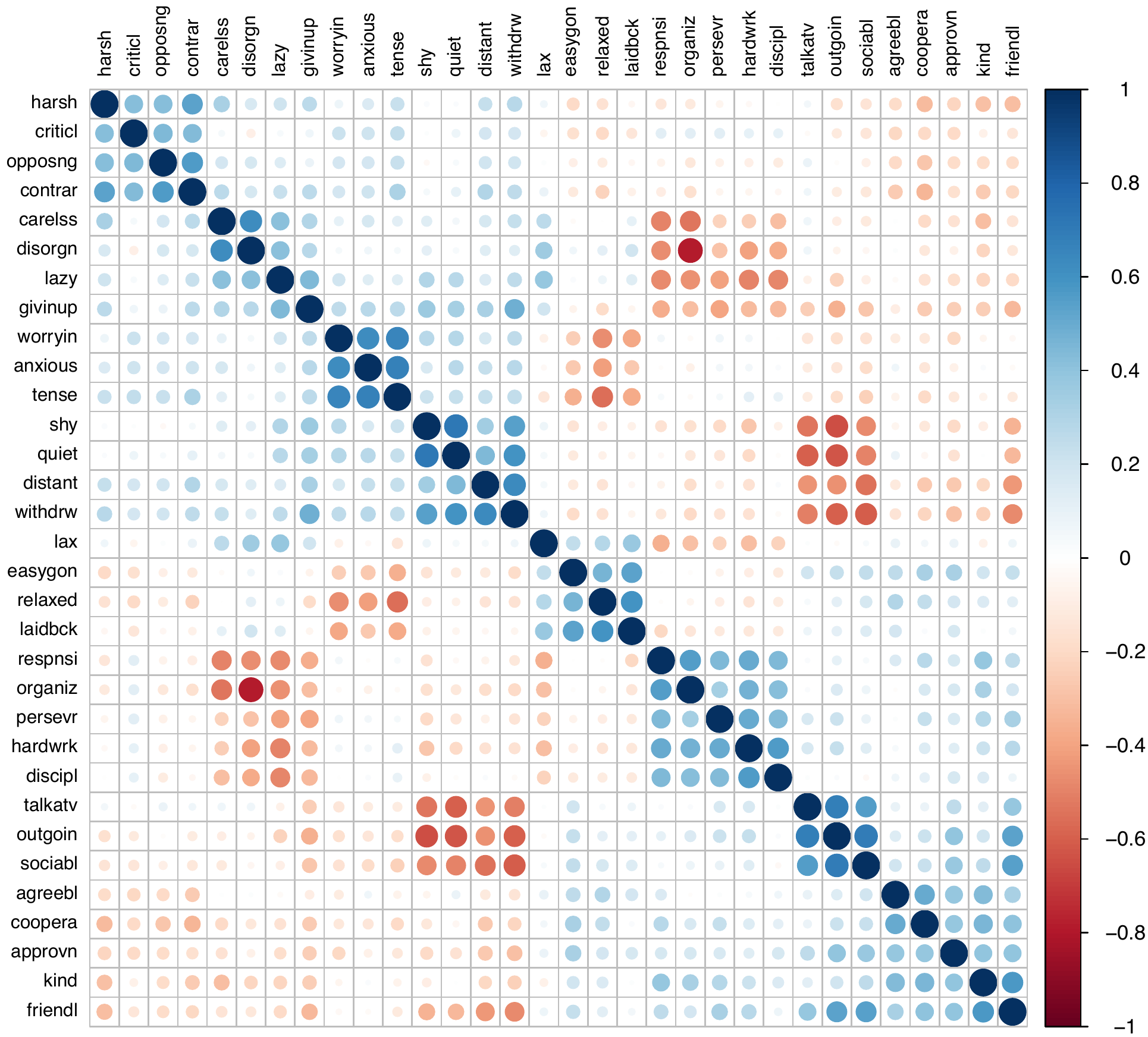}
\caption{Correlation matrix of personality traits from the data set described in~\cite{malle1995puzzle}.}	
\label{fig:FAcorr}
\end{figure}    It appears to have a block structure with predominantly positive entries in each diagonal block and negative entries in the off-diagonal block. Also analyzing the respective variables, they seem to correspond to positive and negative traits. It is therefore natural to assume that this data set follows a signed $\mtp$ distribution and analyze it under this constraint.

The correlation matrix resulting from the sign switching procedure described in~(\ref{eq:Dstar}) is shown on the left in Figure~\ref{fig:FAcorr2}, while the correlation matrix resulting from switching the signs of the 16 (negative) traits that constitute the first block of variables in Figure \ref{fig:FAcorr} is shown on the right in Figure~\ref{fig:FAcorr2}. 
\begin{figure}[t]
\includegraphics[scale=.32]{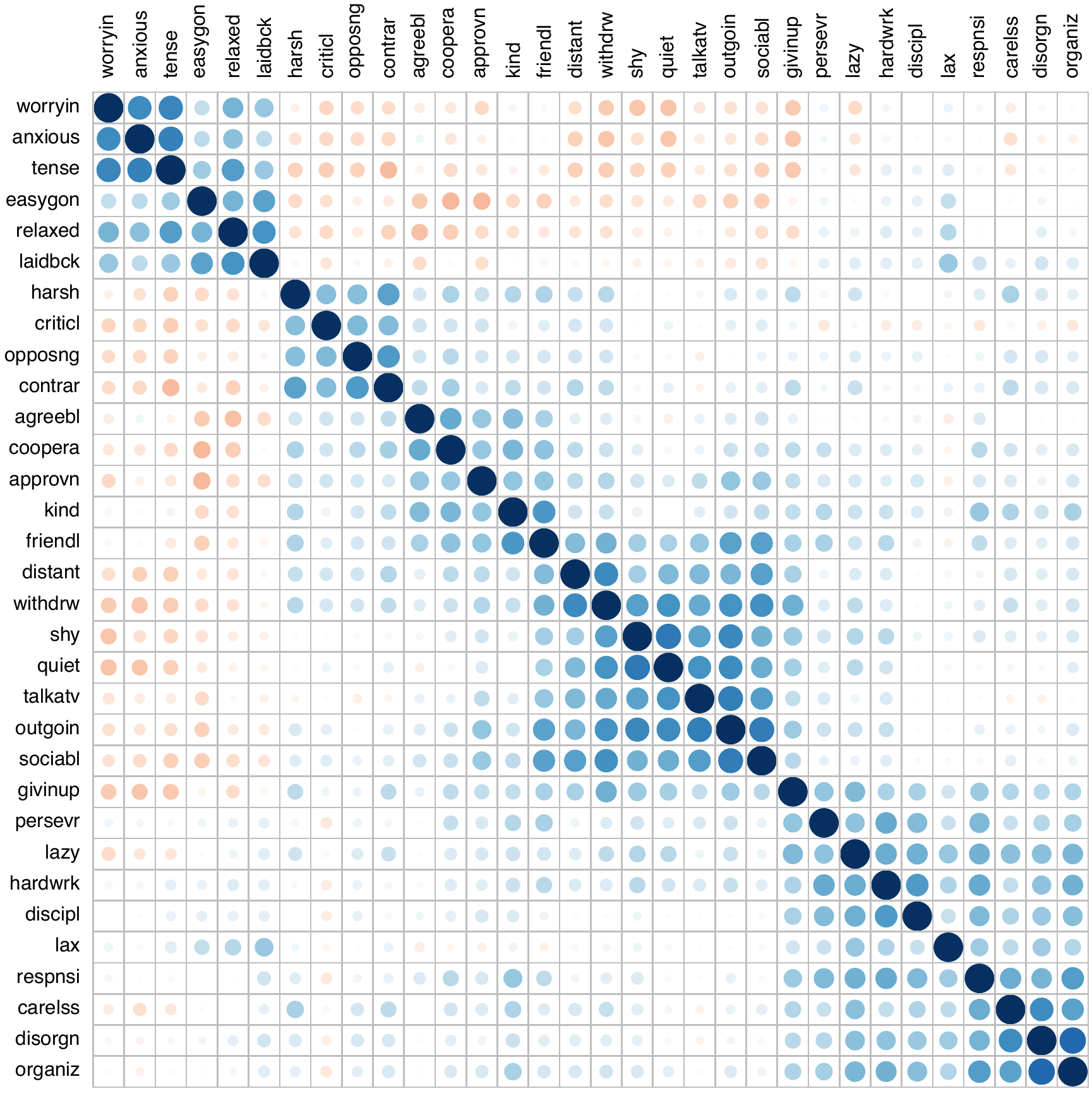}\quad \includegraphics[scale=.32]{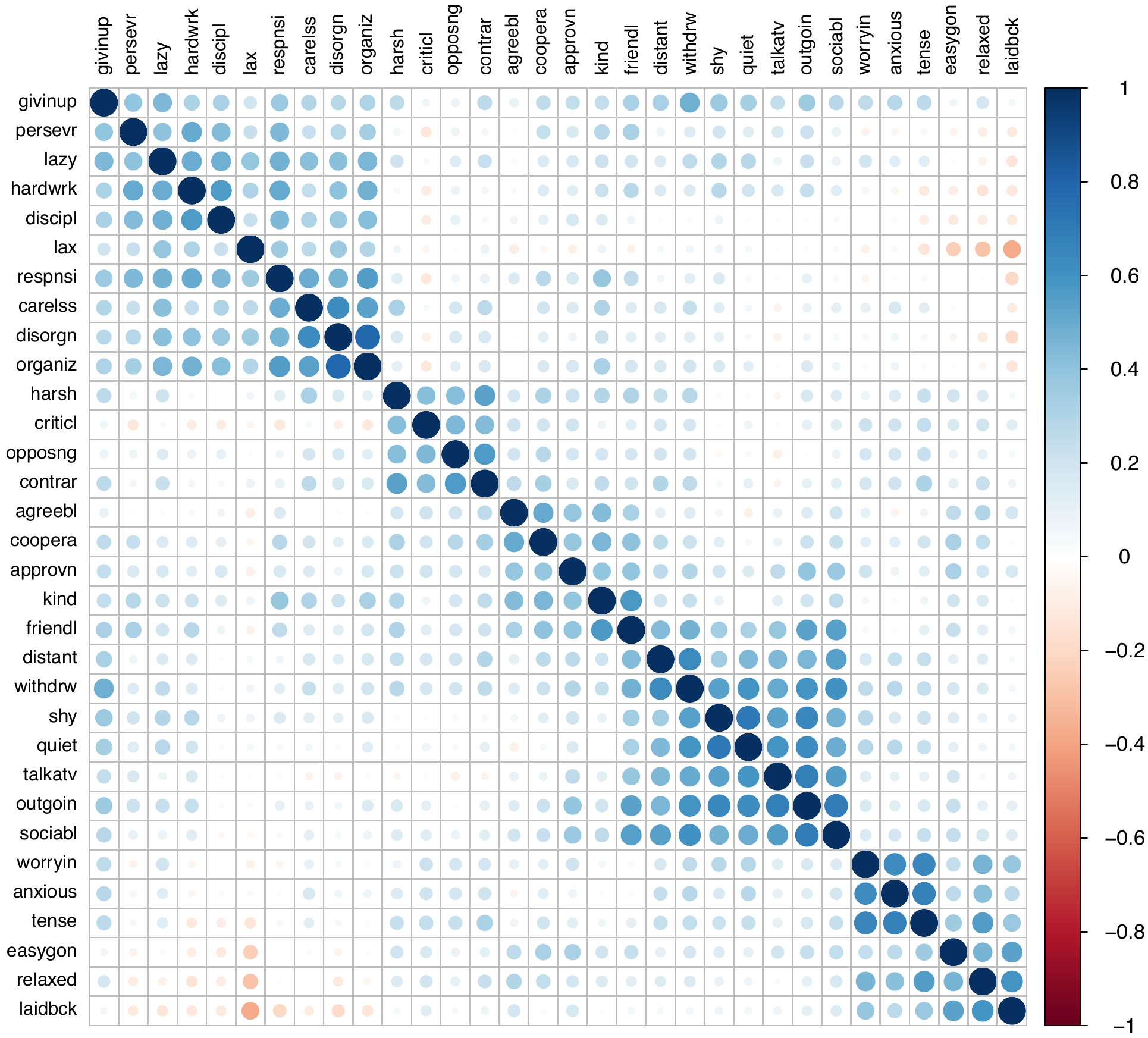}
\caption{The correlation matrix of the data set on personality traits after performing the sign switches as defined in (\ref{eq:Dstar}) is shown on the left. The correlation matrix resulting from switching the signs of the 16 (negative) traits that constitute the first block of variables in Figure \ref{fig:FAcorr} is shown on the right.}	
\label{fig:FAcorr2}
\end{figure}  
These plots suggest that the matrix on the right is closer to being $\mtp$. In fact, its log-likelihood (i.e., the value of $\,\frac{n}{2}(\log \det K-{\rm tr}(SK))$) is -2046.146, as compared to the log-likelihood value of -2071.717 resulting from the sign switching procedure described in~(\ref{eq:Dstar}). For comparison, the value of the unconstrained log-likelihood is -1725.075 and the value of the log-likelihood under $\mtp$ without sign switching is -2356.639. The unconstrained log-likelihood gives a lower bound of 642.142 on the likelihood ratio statistic to test signed $\mtp$ constraints, while the likelihood ratio statistic to test $\mtp$ constraints against the saturated model is equal to 1263.128. 

The graphical models based on no sign switching and switching the signs of the 16 negative traits  are shown in Figure~\ref{fig:FAgraph}. 
\begin{figure}[htb]
\includegraphics[scale=.37]{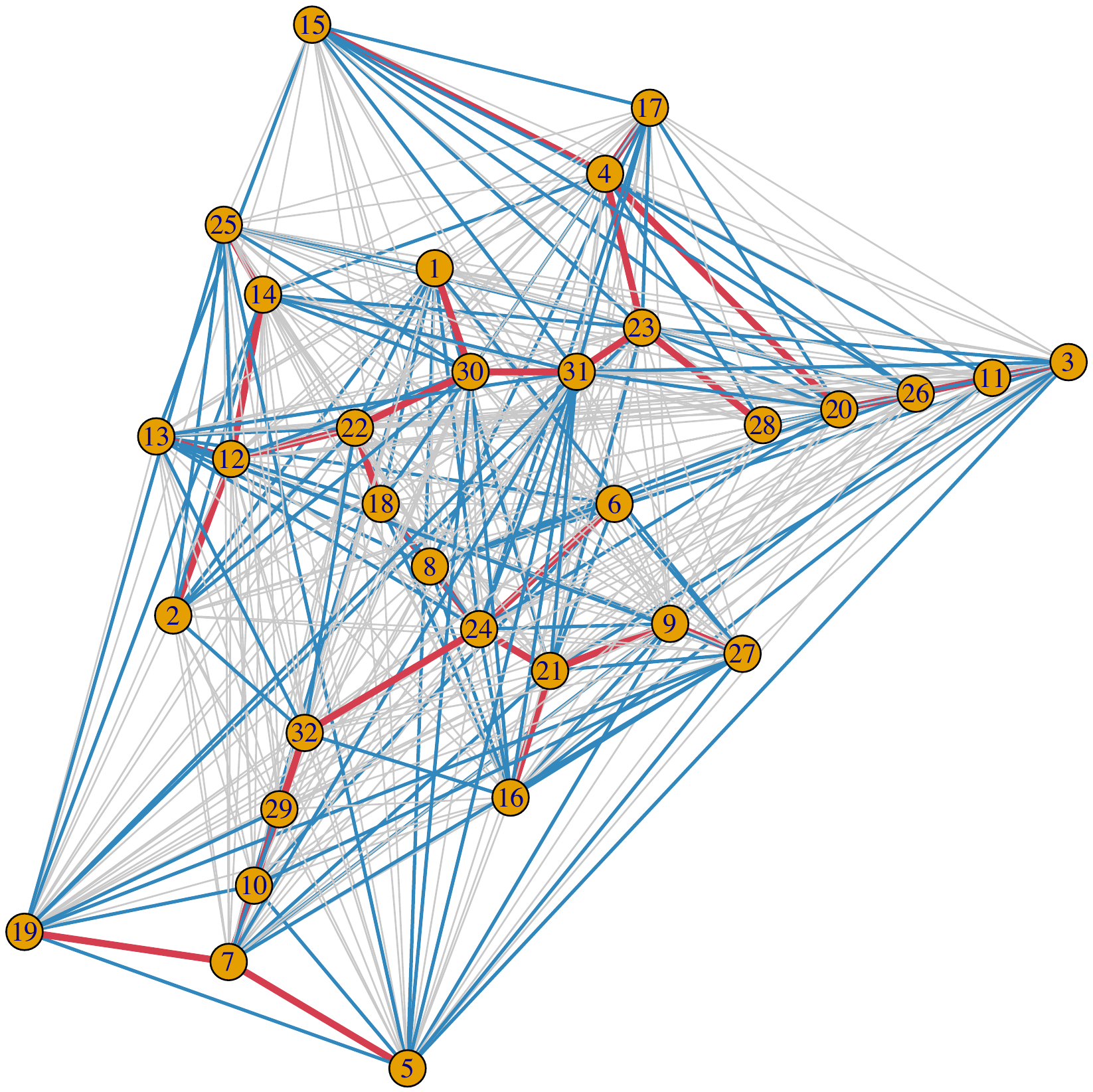}\includegraphics[scale=.37]{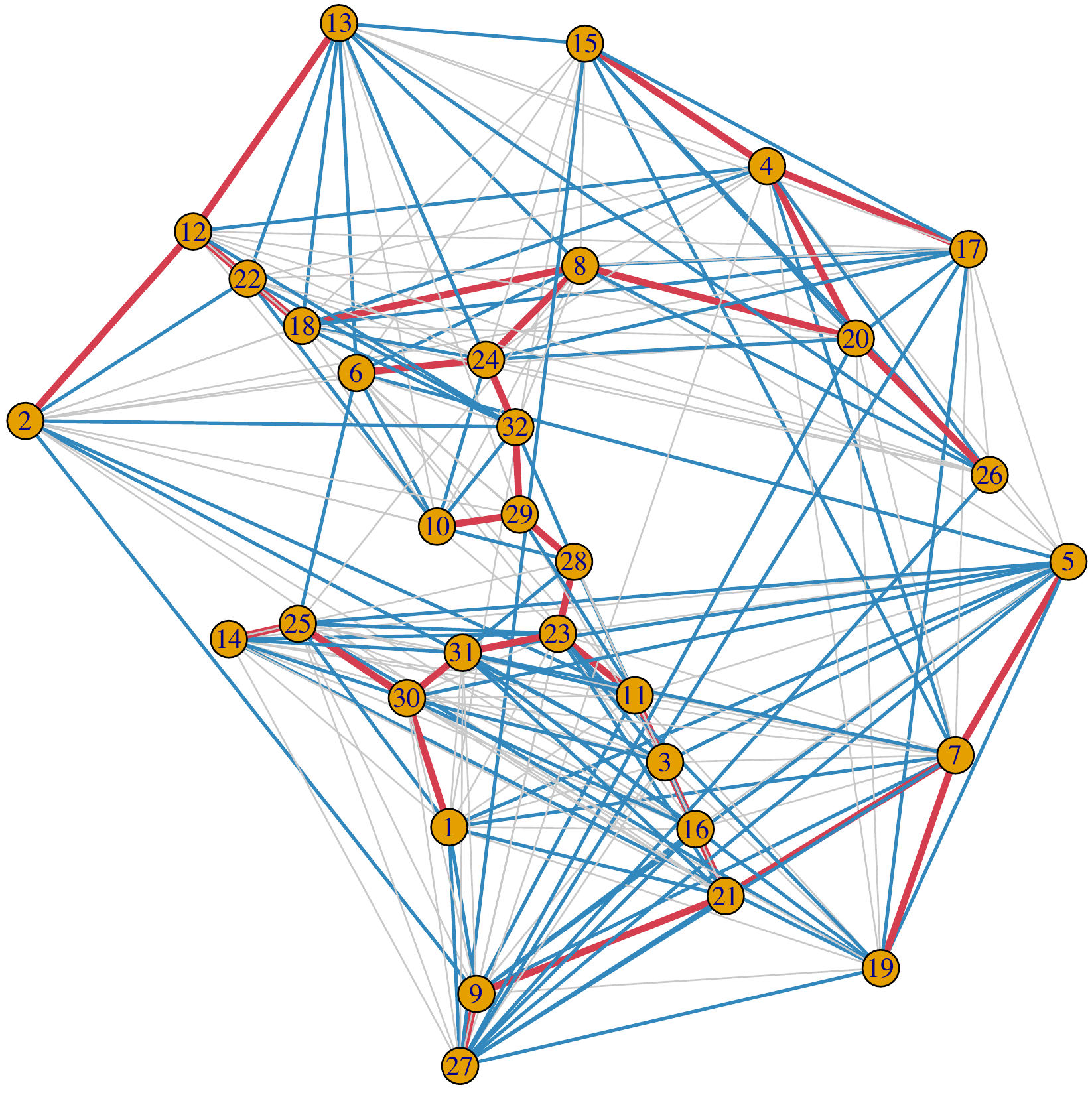}
\caption{On the left, the graphical models resulting from estimation under $\mtp$ based on the correlation matrix shown in Figure~\ref{fig:FAcorr} and, on the right, the correlation matrix shown in Figure~\ref{fig:FAcorr2} (right). The thin gray edges correspond to the edges of the EC graph that are not part of the ML graph. The blue edges represent edges of the ML graph that are not part of the minimum weight spanning tree. The latter is represented by thick red edges.}	
\label{fig:FAgraph}
\end{figure}    
The vertex labels are as shown in Table \ref{labels}. \begin{table}[b]
\centering
\caption{Vertex labeling for Figure \ref{fig:FAgraph}.}
\label{labels}
\begin{tabular}{cccccccc}
\rowcolor[HTML]{EFEFEF} 
1       & 2       & 3       & 4       & 5       & 6       & 7       & 8       \\
distant & talkatv & carelss & hardwrk & anxious & agreebl & tense   & kind    \\
\rowcolor[HTML]{EFEFEF} 
9       & 10      & 11      & 12      & 13      & 14      & 15      & 16      \\
opposng & relaxed & disorgn & outgoin & approvn & shy     & discipl & harsh   \\
\rowcolor[HTML]{EFEFEF} 
17      & 18      & 19      & 20      & 21      & 22      & 23      & 24      \\
persevr & friendl & worryin & respnsi & contrar & sociabl & lazy    & coopera \\
\rowcolor[HTML]{EFEFEF} 
25      & 26      & 27      & 28      & 29      & 30      & 31      & 32      \\
quiet   & organiz & criticl & lax     & laidbck & withdrw & givinup & easygon
\end{tabular}
\end{table}The red edges correspond to the \add{maximum weight} spanning trees. Red and blue edges together form the edge set of the ML graph \add{so in both of these cases we have  $\mwsf(R)\subseteq\hat G$}. Finally, the grey edges are the remaining edges in the EC graph. As expected, the graph on the right looks denser. The interpretation of the spanning tree in both cases is very different. Edges in the first one connect similar personalities such as 6-24 (agreeable and cooperative), 12-22 (outgoing and sociable), 11-23 (disorganized and lazy). On the other hand, the second tree looks similar but it links also some almost perfect opposite personalities such as 12-14 (outgoing and shy),  22-30 (sociable and withdrawn), 11-26 (disorganized and organized), 7-10 (tense and relaxed). Note that none of these four edges are part of the ML graph on the left in Figure~\ref{fig:FAgraph}.
\end{ex}

\section{Discussion}
In this article we have investigated maximum likelihood estimation for Gaussian distributions under the restriction of multivariate total positivity, used a connection to ultrametrics to show that it has a unique solution when the number of observations is at least two, shown that under certain circumstances the MLE can be obtained explicitely, \del
 and given convergent algorithms for calculating the MLE. For signed $\mtp$ distributions we have also given conditions under which a heuristic procedure for applying sign changes is correct and can be used to obtain the MLE.

It remains an issue to consider the asymptotic properties of the estimators we have given, and 
to derive reliable methods for identifying whether a given sample is consistent with the $\mtp$ assumption. On the former issue, standard arguments for convex exponential families ensure that if the true value $K_0$ is an M-matrix, $\hat K$ is a consistent estimator of $K_0$; and this is true whether or not the $\mtp$ assumption is envoked.

 Another question is  whether the ML graph $\hat G$ will be consistent for the true dependence graph. It is clear that without some form of penalty or thresholding, it cannot be the case. For example, if $p=2$ and  the true  $\Sigma$ is a diagonal matrix, the distribution of the empirical correlation $R_{12}$  will be symmetric around $0$. Hence, with probability $1/2$ the ML graph contains an edge between $1$ and $2$ and with probability $1/2$ it does not contain such an edge. This phenomenon persists for any number of observations $n$. Thus, to achieve consistent estimation of the dependence graph of $\Sigma$, some form of penalty for complexity or thresholding must be applied, the latter being suggested by \cite{slawski2015estimation}, who also suggest a refitting after thresholding to ensure positive definiteness of the thresholded matrix. However, positive definiteness is automatically ensured, as shown below.
\begin{prop}Let $K$ be an M-matrix over $V$ and $G=(V,E)$ an undirected graph. Define $K^G$ by \[K^G_{uv}= \begin{cases}K_{uv} &\text{ if } u=v \text{ or } uv\in E\\
0&\text{ otherwise.}
\end{cases}\]
Then $K^G$ is an M-matrix.
\end{prop}
\begin{proof}We may without loss of generality assume that $K$ is scaled such that all diagonal elements are equal to 1; also it is clearly sufficient to consider the case when only a single off-diagonal entry $K_{uv}$ is replaced by zero.  We have to show that the resulting matrix $K^G$ is positive definite.  

Now, let $A=\{u,v\}$ and $B=V\setminus A$ and consider the Schur complements
\[  K/K_{BB} = K_{AA}-K_{AB}(K_{BB})^{-1}K_{BA};\;   K^G/K_{BB}= K_{AA}^G-K_{AB}(K_{BB})^{-1}K_{BA}.\]
Since $K^G_{BB}=K_{BB}$, $K^G$ is positive definite if and only if $K^G/K_{BB}$ is. Because $K$ is an M-matrix, all entries in $K_{AB}(K_{BB})^{-1}K_{BA}$ are non-negative. Hence, we can write the Schur complements as
\[ K/K_{BB} = \begin{pmatrix}1-c&-(a+b)\\-(a+b)&1-d\end{pmatrix};\quad  K^G/K_{BB} = \begin{pmatrix} 1-c&-b\\-b&1-d\end{pmatrix},\]
where $c,d \in (0,1)$ and $a,b\geq 0$.
Since $K$ is positive definite we have 
\[(a+b)^2<(1-c)(1-d)\]
and hence
\[b^2< (1-c)(1-d)-a^2-2ab\leq (1-c)(1-d)\]
implying that $K^G/K_{BB}$ is positive definite. This completes the proof.
\end{proof}

The consistency of the estimator $\hat K$ ensures that the ML graph will eventually contain the true dependence graph when $n$ becomes large and with an appropriate thresholding or penalization, this  
ensures that the true graph can be recovered, as also argued in \cite{slawski2015estimation}.

The issue of the asymptotic distribution of the likelihood ratio test for $\mtp$ is an instance of testing a convex hypothesis within an exponential family of distributions. In our particular case, the convex hypothesis is a polyhedral cone with facets determined by the dependence graph $G(K)$. In such cases, the likelihood ratio test for the convex hypothesis typically has an asymptotic distribution which is a mixture of $\chi^2$-distributions with degrees of freedom determined by the co-dimension of these facets; 
see for example the analysis of  the case of multivariate positivity in models for binary data by \cite{forcina2000}, using results of \cite{shapiro1988tut}.

While these issues are both interesting and important, we consider them to be outside the scope of the present paper as they may be most efficiently dealt with in the more general context of exponential families, containing both the Gaussian and binary cases as special instances. We plan to return to these and other problems in the future.

\section*{Acknowledgements}
Caroline Uhler was partially supported by DARPA (W911NF-16-1-0551), NSF (DMS-1651995), ONR (N00014-17-1-2147), and a Sloan Fellowship. We also thank two anonymous referees for their helpful comments.

\appendix

\begin{thebibliography}{42}

\bibitem[\protect\citeauthoryear{Anandkumar et~al.}{2012}]{anandkumar2012high}
\begin{barticle}[author]
\bauthor{\bsnm{Anandkumar},~\bfnm{Animashree}\binits{A.}},
  \bauthor{\bsnm{Tan},~\bfnm{Vincent~YF}\binits{V.~Y.}},
  \bauthor{\bsnm{Huang},~\bfnm{Furong}\binits{F.}} \AND
  \bauthor{\bsnm{Willsky},~\bfnm{Alan~S}\binits{A.~S.}}
(\byear{2012}).
\btitle{High-dimensional {G}aussian graphical model selection: {W}alk
  summability and local separation criterion}.
\bjournal{Journal of Machine Learning Research}
\bvolume{13}
\bpages{2293--2337}.
\end{barticle}
\endbibitem

\bibitem[\protect\citeauthoryear{Bartolucci and
  Besag}{2002}]{bartolucci2002recursive}
\begin{barticle}[author]
\bauthor{\bsnm{Bartolucci},~\bfnm{Francesco}\binits{F.}} \AND
  \bauthor{\bsnm{Besag},~\bfnm{Julian}\binits{J.}}
(\byear{2002}).
\btitle{A recursive algorithm for {M}arkov random fields}.
\bjournal{Biometrika}
\bvolume{89}
\bpages{724--730}.
\end{barticle}
\endbibitem

\bibitem[\protect\citeauthoryear{Bartolucci and Forcina}{2000}]{forcina2000}
\begin{barticle}[author]
\bauthor{\bsnm{Bartolucci},~\bfnm{Francesco}\binits{F.}} \AND
  \bauthor{\bsnm{Forcina},~\bfnm{Antonio}\binits{A.}}
(\byear{2000}).
\btitle{A likelihood ratio test for {$\rm MTP\sb 2$} within binary variables}.
\bjournal{Annals of Statistics}
\bvolume{28}
\bpages{1206--1218}.
\end{barticle}
\endbibitem

\bibitem[\protect\citeauthoryear{Bhattacharya}{2012}]{bhattacharya2012covariance}
\begin{barticle}[author]
\bauthor{\bsnm{Bhattacharya},~\bfnm{Bhaskar}\binits{B.}}
(\byear{2012}).
\btitle{Covariance selection and multivariate dependence}.
\bjournal{Journal of Multivariate Analysis}
\bvolume{106}
\bpages{212--228}.
\end{barticle}
\endbibitem

\bibitem[\protect\citeauthoryear{B{\o}lviken}{1982}]{B}
\begin{barticle}[author]
\bauthor{\bsnm{B{\o}lviken},~\bfnm{Erik}\binits{E.}}
(\byear{1982}).
\btitle{Probability Inequalities for the Multivariate Normal with Non-negative
  Partial Correlations}.
\bjournal{Scandinavian Journal of Statistics}
\bvolume{9}
\bpages{49--58}.
\end{barticle}
\endbibitem

\bibitem[\protect\citeauthoryear{Buhl}{1993}]{buhl}
\begin{barticle}[author]
\bauthor{\bsnm{Buhl},~\bfnm{S{\o}ren~L.}\binits{S.~L.}}
(\byear{1993}).
\btitle{On the existence of maximum likelihood estimators for graphical
  {G}aussian models}.
\bjournal{Scandinavian Journal of Statistics}
\bvolume{20}
\bpages{263--270}.
\end{barticle}
\endbibitem

\bibitem[\protect\citeauthoryear{Choi et~al.}{2011}]{willsky2011}
\begin{barticle}[author]
\bauthor{\bsnm{Choi},~\bfnm{Myung~Jin}\binits{M.~J.}},
  \bauthor{\bsnm{Tan},~\bfnm{Vincent Y.~F.}\binits{V.~Y.~F.}},
  \bauthor{\bsnm{Anandkumar},~\bfnm{Animashree}\binits{A.}} \AND
  \bauthor{\bsnm{Willsky},~\bfnm{Alan~S.}\binits{A.~S.}}
(\byear{2011}).
\btitle{Learning latent tree graphical models}.
\bjournal{Journal of Machine Learning Research}
\bvolume{12}
\bpages{1771--1812}.
\end{barticle}
\endbibitem

\bibitem[\protect\citeauthoryear{Colangelo, Scarsini and
  Shaked}{2005}]{colangelo2005some}
\begin{barticle}[author]
\bauthor{\bsnm{Colangelo},~\bfnm{Antonio}\binits{A.}},
  \bauthor{\bsnm{Scarsini},~\bfnm{Marco}\binits{M.}} \AND
  \bauthor{\bsnm{Shaked},~\bfnm{Moshe}\binits{M.}}
(\byear{2005}).
\btitle{Some notions of multivariate positive dependence}.
\bjournal{Insurance: Mathematics and Economics}
\bvolume{37}
\bpages{13--26}.
\end{barticle}
\endbibitem

\bibitem[\protect\citeauthoryear{Dellacherie, Martinez and
  San~Martin}{2014}]{dellacherie2014inverse}
\begin{bbook}[author]
\bauthor{\bsnm{Dellacherie},~\bfnm{Claude}\binits{C.}},
  \bauthor{\bsnm{Martinez},~\bfnm{Servet}\binits{S.}} \AND
  \bauthor{\bsnm{San~Martin},~\bfnm{Jaime}\binits{J.}}
(\byear{2014}).
\btitle{Inverse M-matrices and ultrametric matrices}
\bvolume{2118}.
\bpublisher{Springer}.
\end{bbook}
\endbibitem

\bibitem[\protect\citeauthoryear{Dempster}{1972}]{dempster1972covariance}
\begin{barticle}[author]
\bauthor{\bsnm{Dempster},~\bfnm{Arthur~P}\binits{A.~P.}}
(\byear{1972}).
\btitle{Covariance selection}.
\bjournal{Biometrics}
\bpages{157--175}.
\end{barticle}
\endbibitem

\bibitem[\protect\citeauthoryear{Djolonga and
  Krause}{2015}]{djolonga2015scalable}
\begin{barticle}[author]
\bauthor{\bsnm{Djolonga},~\bfnm{Josip}\binits{J.}} \AND
  \bauthor{\bsnm{Krause},~\bfnm{Andreas}\binits{A.}}
(\byear{2015}).
\btitle{Scalable Variational Inference in Log-supermodular Models}.
\bjournal{In International Conference on Machine Learning (ICML)}.
\end{barticle}
\endbibitem

\bibitem[\protect\citeauthoryear{Egilmez, Pavez and
  Ortega}{2016}]{egilmez:etal:16}
\begin{bmisc}[author]
\bauthor{\bsnm{Egilmez},~\bfnm{H.~E.}\binits{H.~E.}},
  \bauthor{\bsnm{Pavez},~\bfnm{E.}\binits{E.}} \AND
  \bauthor{\bsnm{Ortega},~\bfnm{A.}\binits{A.}}
(\byear{2016}).
\btitle{Graph Learning from Data under Structural and Laplacian constraints}.
\bhowpublished{arXiv:1611.0518}.
\end{bmisc}
\endbibitem

\bibitem[\protect\citeauthoryear{Fallat et~al.}{2017}]{MTP2Markov2015}
\begin{barticle}[author]
\bauthor{\bsnm{Fallat},~\bfnm{Shaun}\binits{S.}},
  \bauthor{\bsnm{Lauritzen},~\bfnm{Steffen~L.}\binits{S.~L.}},
  \bauthor{\bsnm{Sadeghi},~\bfnm{Kayvan}\binits{K.}},
  \bauthor{\bsnm{Uhler},~\bfnm{Caroline}\binits{C.}},
  \bauthor{\bsnm{Wermuth},~\bfnm{Nanny}\binits{N.}} \AND
  \bauthor{\bsnm{Zwiernik},~\bfnm{Piotr}\binits{P.}}
(\byear{2017}).
\btitle{Total positivity in {M}arkov structures}.
\bjournal{Annals of Statistics}
\bvolume{45}
\bpages{1152--1184}.
\end{barticle}
\endbibitem

\bibitem[\protect\citeauthoryear{Felsenstein}{1973}]{felsenstein_maximum-likelihood_1973}
\begin{barticle}[author]
\bauthor{\bsnm{Felsenstein},~\bfnm{Joseph}\binits{J.}}
(\byear{1973}).
\btitle{Maximum-likelihood estimation of evolutionary trees from continuous
  characters.}
\bjournal{American Journal of Human Genetics}
\bvolume{25}
\bpages{471--492}.
\end{barticle}
\endbibitem

\bibitem[\protect\citeauthoryear{Fortuin, Kasteleyn and
  Ginibre}{1971}]{fortuin1971correlation}
\begin{barticle}[author]
\bauthor{\bsnm{Fortuin},~\bfnm{Cees~M}\binits{C.~M.}},
  \bauthor{\bsnm{Kasteleyn},~\bfnm{Pieter~W}\binits{P.~W.}} \AND
  \bauthor{\bsnm{Ginibre},~\bfnm{Jean}\binits{J.}}
(\byear{1971}).
\btitle{Correlation inequalities on some partially ordered sets}.
\bjournal{Communications of Mathematical Physics}
\bvolume{22}
\bpages{89--103}.
\end{barticle}
\endbibitem

\bibitem[\protect\citeauthoryear{Gower and Ross}{1969}]{gower1969minimum}
\begin{barticle}[author]
\bauthor{\bsnm{Gower},~\bfnm{John~C}\binits{J.~C.}} \AND
  \bauthor{\bsnm{Ross},~\bfnm{GJS}\binits{G.}}
(\byear{1969}).
\btitle{Minimum spanning trees and single linkage cluster analysis}.
\bjournal{Applied Statistics}
\bpages{54--64}.
\end{barticle}
\endbibitem

\bibitem[\protect\citeauthoryear{Grant and Boyd}{2014}]{cvx}
\begin{bmisc}[author]
\bauthor{\bsnm{Grant},~\bfnm{Michael}\binits{M.}} \AND
  \bauthor{\bsnm{Boyd},~\bfnm{Stephen}\binits{S.}}
(\byear{2014}).
\btitle{{CVX}: Matlab Software for Disciplined Convex Programming, version
  2.1}.
\bhowpublished{\url{http://cvxr.com/cvx}}.
\end{bmisc}
\endbibitem

\bibitem[\protect\citeauthoryear{Gross and
  Sullivant}{2017}]{gross:sullivant:17}
\begin{barticle}[author]
\bauthor{\bsnm{Gross},~\bfnm{E.}\binits{E.}} \AND
  \bauthor{\bsnm{Sullivant},~\bfnm{S.}\binits{S.}}
(\byear{2017}).
\btitle{The Maximum Likelihood Threshold of a Graph}.
\bjournal{Bernoulli}.
\bnote{To appear}.
\end{barticle}
\endbibitem

\bibitem[\protect\citeauthoryear{H{\o}jsgaard, Edwards and
  Lauritzen}{2012}]{GraphModelR}
\begin{bbook}[author]
\bauthor{\bsnm{H{\o}jsgaard},~\bfnm{S.}\binits{S.}},
  \bauthor{\bsnm{Edwards},~\bfnm{D.}\binits{D.}} \AND
  \bauthor{\bsnm{Lauritzen},~\bfnm{S.}\binits{S.}}
(\byear{2012}).
\btitle{Graphical {M}odels with {R}}.
\bpublisher{Springer}, \baddress{New York}.
\end{bbook}
\endbibitem

\bibitem[\protect\citeauthoryear{Johnson and Smith}{1996}]{johnson:smith:96}
\begin{barticle}[author]
\bauthor{\bsnm{Johnson},~\bfnm{Charles~R.}\binits{C.~R.}} \AND
  \bauthor{\bsnm{Smith},~\bfnm{Ronald~L.}\binits{R.~L.}}
(\byear{1996}).
\btitle{The Completion Problem for $M$-Matrices and Inverse $M$-Matrices}.
\bjournal{Linear Algebra and Its Applications}
\bvolume{241--243}
\bpages{655-667}.
\end{barticle}
\endbibitem

\bibitem[\protect\citeauthoryear{Johnson and Smith}{1999}]{johnson1999path}
\begin{barticle}[author]
\bauthor{\bsnm{Johnson},~\bfnm{Charles~R}\binits{C.~R.}} \AND
  \bauthor{\bsnm{Smith},~\bfnm{Ronald~L}\binits{R.~L.}}
(\byear{1999}).
\btitle{Path product matrices}.
\bjournal{Linear and Multilinear Algebra}
\bvolume{46}
\bpages{177--191}.
\end{barticle}
\endbibitem

\bibitem[\protect\citeauthoryear{Johnson and Smith}{2011}]{johnson2011inverse}
\begin{barticle}[author]
\bauthor{\bsnm{Johnson},~\bfnm{Charles~R}\binits{C.~R.}} \AND
  \bauthor{\bsnm{Smith},~\bfnm{Ronald~L}\binits{R.~L.}}
(\byear{2011}).
\btitle{{Inverse M-matrices, II}}.
\bjournal{Linear Algebra and its Applications}
\bvolume{435}
\bpages{953--983}.
\end{barticle}
\endbibitem

\bibitem[\protect\citeauthoryear{Karlin and Rinott}{1980}]{KarlinRinott80}
\begin{barticle}[author]
\bauthor{\bsnm{Karlin},~\bfnm{Samuel}\binits{S.}} \AND
  \bauthor{\bsnm{Rinott},~\bfnm{Yosef}\binits{Y.}}
(\byear{1980}).
\btitle{Classes of orderings of measures and related correlation inequalities.
  {I}. {M}ultivariate totally positive distributions}.
\bjournal{Journal of Multivariate Analysis}
\bvolume{10}
\bpages{467--498}.
\end{barticle}
\endbibitem

\bibitem[\protect\citeauthoryear{Karlin and Rinott}{1981}]{karlin1981signed}
\begin{barticle}[author]
\bauthor{\bsnm{Karlin},~\bfnm{Samuel}\binits{S.}} \AND
  \bauthor{\bsnm{Rinott},~\bfnm{Yosef}\binits{Y.}}
(\byear{1981}).
\btitle{Total Positivity Properties of Absolute Value Multinormal Variables
  with Applications to Confidence Interval Estimates and Related Probabilistic
  Inequalities}.
\bjournal{Annals of Statistics}
\bvolume{9}
\bpages{1035--1049}.
\end{barticle}
\endbibitem

\bibitem[\protect\citeauthoryear{Karlin and Rinott}{1983}]{karlinGaussian}
\begin{barticle}[author]
\bauthor{\bsnm{Karlin},~\bfnm{Samuel}\binits{S.}} \AND
  \bauthor{\bsnm{Rinott},~\bfnm{Yosef}\binits{Y.}}
(\byear{1983}).
\btitle{M-Matrices as covariance matrices of multinormal distributions}.
\bjournal{Linear Algebra and its Applications}
\bvolume{52}
\bpages{419 - 438}.
\end{barticle}
\endbibitem

\bibitem[\protect\citeauthoryear{Lauritzen}{1996}]{lau96}
\begin{bbook}[author]
\bauthor{\bsnm{Lauritzen},~\bfnm{S.~L.}\binits{S.~L.}}
(\byear{1996}).
\btitle{Graphical Models}.
\bpublisher{Clarendon Press}, \baddress{Oxford, United Kingdom}.
\end{bbook}
\endbibitem

\bibitem[\protect\citeauthoryear{Ledermann}{1940}]{ledermann1940problem}
\begin{barticle}[author]
\bauthor{\bsnm{Ledermann},~\bfnm{Walter}\binits{W.}}
(\byear{1940}).
\btitle{{I}.---{O}n a Problem concerning Matrices with Variable Diagonal
  Elements.}
\bjournal{Proceedings of the Royal Society of Edinburgh}
\bvolume{60}
\bpages{1--17}.
\end{barticle}
\endbibitem

\bibitem[\protect\citeauthoryear{Luo and Tseng}{1992}]{tseng1992}
\begin{barticle}[author]
\bauthor{\bsnm{Luo},~\bfnm{Z.~Q.}\binits{Z.~Q.}} \AND
  \bauthor{\bsnm{Tseng},~\bfnm{P.}\binits{P.}}
(\byear{1992}).
\btitle{On the convergence of the coordinate descent method for convex
  differentiable minimization}.
\bjournal{Journal of Optimization Theory and Applications}
\bvolume{72}
\bpages{7--35}.
\end{barticle}
\endbibitem

\bibitem[\protect\citeauthoryear{Malioutov, Johnson and
  Willsky}{2006}]{malioutov2006walk}
\begin{barticle}[author]
\bauthor{\bsnm{Malioutov},~\bfnm{Dmitry~M}\binits{D.~M.}},
  \bauthor{\bsnm{Johnson},~\bfnm{Jason~K}\binits{J.~K.}} \AND
  \bauthor{\bsnm{Willsky},~\bfnm{Alan~S}\binits{A.~S.}}
(\byear{2006}).
\btitle{Walk-sums and belief propagation in {G}aussian graphical models}.
\bjournal{Journal of Machine Learning Research}
\bvolume{7}
\bpages{2031--2064}.
\end{barticle}
\endbibitem

\bibitem[\protect\citeauthoryear{Malle and Horowitz}{1995}]{malle1995puzzle}
\begin{barticle}[author]
\bauthor{\bsnm{Malle},~\bfnm{Bertram~F}\binits{B.~F.}} \AND
  \bauthor{\bsnm{Horowitz},~\bfnm{Leonard~M}\binits{L.~M.}}
(\byear{1995}).
\btitle{The puzzle of negative self-views: An exploration using the schema
  concept.}
\bjournal{Journal of Personality and Social Psychology}
\bvolume{68}
\bpages{470}.
\end{barticle}
\endbibitem

\bibitem[\protect\citeauthoryear{Newman}{1983}]{newman1983general}
\begin{barticle}[author]
\bauthor{\bsnm{Newman},~\bfnm{Charles~M}\binits{C.~M.}}
(\byear{1983}).
\btitle{A general central limit theorem for {FKG} systems}.
\bjournal{Communications of Mathematical Physics}
\bvolume{91}
\bpages{75--80}.
\end{barticle}
\endbibitem

\bibitem[\protect\citeauthoryear{Ostrowski}{1937}]{O}
\begin{barticle}[author]
\bauthor{\bsnm{Ostrowski},~\bfnm{Alexander}\binits{A.}}
(\byear{1937}).
\btitle{{\"U}ber die {D}eterminanten mit {\"u}berwiegender {H}auptdiagonale}.
\bjournal{Commentarii Mathematici Helvetici}
\bvolume{10}
\bpages{69--96}.
\end{barticle}
\endbibitem

\bibitem[\protect\citeauthoryear{Propp and Wilson}{1996}]{propp1996exact}
\begin{barticle}[author]
\bauthor{\bsnm{Propp},~\bfnm{James~Gary}\binits{J.~G.}} \AND
  \bauthor{\bsnm{Wilson},~\bfnm{David~Bruce}\binits{D.~B.}}
(\byear{1996}).
\btitle{Exact sampling with coupled {M}arkov chains and applications to
  statistical mechanics}.
\bjournal{Random Structures and Algorithms}
\bvolume{9}
\bpages{223--252}.
\end{barticle}
\endbibitem

\bibitem[\protect\citeauthoryear{Semple and
  Steel}{2003}]{semple2003phylogenetics}
\begin{bbook}[author]
\bauthor{\bsnm{Semple},~\bfnm{Charles}\binits{C.}} \AND
  \bauthor{\bsnm{Steel},~\bfnm{Mike~A}\binits{M.~A.}}
(\byear{2003}).
\btitle{Phylogenetics}
\bvolume{24}.
\bpublisher{Oxford University Press}.
\end{bbook}
\endbibitem

\bibitem[\protect\citeauthoryear{Shapiro}{1988}]{shapiro1988tut}
\begin{barticle}[author]
\bauthor{\bsnm{Shapiro},~\bfnm{A.}\binits{A.}}
(\byear{1988}).
\btitle{{Towards a unified theory of inequality constrained testing in
  multivariate analysis}}.
\bjournal{International Statistical Review}
\bvolume{56}
\bpages{49--62}.
\end{barticle}
\endbibitem

\bibitem[\protect\citeauthoryear{Shiers et~al.}{2016}]{ASSZ2014}
\begin{barticle}[author]
\bauthor{\bsnm{Shiers},~\bfnm{Nathaniel}\binits{N.}},
  \bauthor{\bsnm{Zwiernik},~\bfnm{Piotr}\binits{P.}},
  \bauthor{\bsnm{Aston},~\bfnm{John}\binits{J.}} \AND
  \bauthor{\bsnm{Smith},~\bfnm{Jim~Q.}\binits{J.~Q.}}
(\byear{2016}).
\btitle{The correlation space of {G}aussian latent tree models and model
  selection without fitting}.
\bjournal{Biometrika}
\bvolume{103}
\bpages{531--545}.
\end{barticle}
\endbibitem

\bibitem[\protect\citeauthoryear{Slawski and
  Hein}{2015}]{slawski2015estimation}
\begin{barticle}[author]
\bauthor{\bsnm{Slawski},~\bfnm{Martin}\binits{M.}} \AND
  \bauthor{\bsnm{Hein},~\bfnm{Matthias}\binits{M.}}
(\byear{2015}).
\btitle{Estimation of positive definite {M}-matrices and structure learning for
  attractive {G}aussian {M}arkov random fields}.
\bjournal{Linear Algebra and its Applications}
\bvolume{473}
\bpages{145--179}.
\end{barticle}
\endbibitem

\bibitem[\protect\citeauthoryear{Spearman}{1928}]{Spearman:1928qy}
\begin{barticle}[author]
\bauthor{\bsnm{Spearman},~\bfnm{Charles}\binits{C.}}
(\byear{1928}).
\btitle{The Abilities of Man}.
\bjournal{Science}
\bvolume{68}
\bpages{38}.
\end{barticle}
\endbibitem

\bibitem[\protect\citeauthoryear{Speed and Kiiveri}{1986}]{SK1986}
\begin{barticle}[author]
\bauthor{\bsnm{Speed},~\bfnm{T.~P.}\binits{T.~P.}} \AND
  \bauthor{\bsnm{Kiiveri},~\bfnm{H.~T.}\binits{H.~T.}}
(\byear{1986}).
\btitle{Gaussian {M}arkov distributions over finite graphs}.
\bjournal{Annals of Statistics}
\bvolume{14}
\bpages{138--150}.
\end{barticle}
\endbibitem

\bibitem[\protect\citeauthoryear{Uhler}{2012}]{uhler2010}
\begin{barticle}[author]
\bauthor{\bsnm{Uhler},~\bfnm{Caroline}\binits{C.}}
(\byear{2012}).
\btitle{Geometry of maximum likelihood estimation in {G}aussian graphical
  models}.
\bjournal{Annals of Statistics}
\bvolume{40}
\bpages{238--261}.
\end{barticle}
\endbibitem

\bibitem[\protect\citeauthoryear{Wermuth and Scheidt}{1977}]{WS1977}
\begin{barticle}[author]
\bauthor{\bsnm{Wermuth},~\bfnm{Nanny}\binits{N.}} \AND
  \bauthor{\bsnm{Scheidt},~\bfnm{Eberhard}\binits{E.}}
(\byear{1977}).
\btitle{Algorithm {AS} 105: Fitting a Covariance Selection Model to a Matrix}.
\bjournal{Journal of the Royal Statistical Society. Series C (Applied
  Statistics)}
\bvolume{26}
\bpages{pp. 88-92}.
\end{barticle}
\endbibitem

\bibitem[\protect\citeauthoryear{Zwiernik}{2015}]{LTbook}
\begin{bbook}[author]
\bauthor{\bsnm{Zwiernik},~\bfnm{Piotr}\binits{P.}}
(\byear{2015}).
\btitle{{Semialgebraic Statistics and Latent Tree Models}}.
\bseries{Monographs on Statistics and Applied Probability}
\bvolume{146}.
\bpublisher{Chapman \& Hall}.
\end{bbook}
\endbibitem

\end{thebibliography}
\end{document}